\documentclass[10pt,twocolumn]{article}

\usepackage[utf8]{inputenc}
\usepackage[T1]{fontenc}
\usepackage{amsmath}
\usepackage{amssymb} 
\usepackage{amsfonts}
\usepackage{amsthm}

\usepackage[varg]{newtxmath} 

\usepackage[margin=0.75in]{geometry}
\usepackage{microtype}
\usepackage{graphicx}
\usepackage{booktabs}
\usepackage{xcolor}
\usepackage{hyperref}
\hypersetup{colorlinks=true, linkcolor=blue, citecolor=blue, urlcolor=blue}

\usepackage[numbers,sort&compress]{natbib}

\usepackage[blocks]{authblk}

\usepackage{mathtools}
\usepackage{enumitem}
\usepackage[capitalize,noabbrev]{cleveref}

\theoremstyle{plain}
\newtheorem{theorem}{Theorem}[section]
\newtheorem{proposition}[theorem]{Proposition}
\newtheorem{lemma}[theorem]{Lemma}

\theoremstyle{definition}

\theoremstyle{remark}

\newcommand{\R}{\mathbb{R}}

\newcommand{\Var}{\mathrm{Var}}
\newcommand{\ket}[1]{|#1\rangle}
\newcommand{\bra}[1]{\langle#1|}

\title{\Large \textbf{Wavelet Variance Equipartition as a Threshold for\\ World-Model Quality and Quantum Kernel TN-Simulability}}

\author[1,2]{\textbf{Chon-Fai Kam}\thanks{Corresponding author: kam.chonfai@gmail.com}}
\author[3]{\textbf{Xavier Cadet}}
\author[4]{\textbf{Miloud Bessafi}}
\author[1,5]{\textbf{Frederic Cadet}\thanks{frederic.cadet.run@gmail.com}}

\affil[1]{\small University Paris City \& University of Reunion, Paris, France}
\affil[2]{\small Dipartimento di Fisica e Chimica Emilio Segrè, Università degli Studi di Palermo, Palermo, Italy}
\affil[3]{\small Thayer School of Engineering, Dartmouth College, Hanover, NH 03755, USA}
\affil[4]{\small EnergyLab, University of Reunion, Saint-Denis, France}
\affil[5]{\small PEACCEL, AI for Biologics, Paris, France}

\date{}

\begin{document}
\maketitle

\begin{abstract}
While world models excel at learning compact representations of complex environments, they lack a principled, physics-grounded metric for assessing the structural fidelity of their latent spaces. We identify the wavelet scaling exponent $\alpha$ as a critical diagnostic for this structural regularity, proposing that optimal representations should satisfy \emph{variance equipartition} ($\alpha \approx 1/2$)---a condition mirroring the constant energy flux within Kolmogorov's inertial range. We formally establish $\alpha = 1/2$ as a sharp transition boundary for the classical simulability of amplitude-encoded quantum kernels. Using tensor-network theory, we prove that latents with $\alpha > 1/2$ reside in an area-law phase admitting efficient classical emulation, whereas $\alpha < 1/2$ triggers a volume-law phase where the required Matrix Product State bond dimension $\chi$ grows exponentially with the qubit count $n$. Empirical analysis of pre-trained VideoMAE latents reveals a fundamental dichotomy: while spatial token sequences approach the physical equipartition limit ($\alpha \approx 0.423$), the permutation-invariant feature channels exhibit unstructured disorder ($\alpha \approx -0.123$). This channel-wise complexity forces real-world latents deep into the volume-law phase, providing a data-driven necessary structural condition for tensor-network simulation hardness. Finally, to rigorously quantify the measurement overhead of high-dimensional quantum representations, we apply Weingarten calculus to derive the exact analytical variance of the scrambled transition probability under a 2-design ensemble. We prove that this variance scales strictly as $\Var[X] = \Theta(d^{-2})$. We confirm this exact scaling numerically with a log-log slope of $-1.881$ ($R^2 = 0.999$), explicitly identifying a formidable ``shot-noise wall'' demanding a critical measurement budget of $M = \Omega(d^2)$ that constrains the scalability of quantum machine learning.
\end{abstract}

\section{Introduction}
\label{sec:intro}

World models---systems that learn compact, persistent representations of physical environments---have emerged as a central paradigm in machine learning \citep{lecun2022path,bruce2024genie}. Architectures such as the Joint Embedding Predictive Architecture (JEPA) \citep{assran2023self} produce latent vectors $z \in \mathbb{R}^d$ that encode multi-scale structure spanning local geometric features to global causal invariants, without requiring pixel-level reconstruction. Despite these impressive capabilities, evaluating the quality of such representations remains a fundamental challenge: current metrics are task-specific and dataset-dependent, and provide no insight into whether the internal representation has captured the hierarchical, scale-invariant organisation that makes the physical world predictable.

We propose a physics-grounded criterion to fill this gap. Physical environments at high Reynolds numbers are governed by Navier--Stokes dynamics whose solutions exhibit a characteristic multi-scale energy cascade. In Kolmogorov's inertial range \citep{kolmogorov1941local}, energy transfers from large to small scales with constant flux, producing \emph{variance equipartition} across scales when viewed in a multi-resolution basis. We propose that a world model has genuinely internalised physical structure when its latent vectors exhibit the same scaling signature: detail coefficients at dyadic scale $k$ decay as $2^{-\alpha k}$ with $\alpha \approx 1/2$. This is not merely an analogy---we show in Section~\ref{sec:background} that $\alpha = 1/2$ uniquely minimises the integrated squared bias of a downstream linear probe, providing a rigorous optimality statement alongside the physical intuition.

This threshold emerges as a fundamental phase boundary from the perspective of quantum complexity. Quantum kernel methods offer a route to exploit high-dimensional latent structure by amplitude-encoding $z$ into $n = \lceil\log_2 d\rceil$ qubits. We establish $\alpha = 1/2$ as the entanglement phase transition: when $\alpha < 1/2$, the singular value spectrum decays slowly enough to force a \emph{volume-law} phase, rendering classical tensor-network (TN) simulation exponentially hard in the qubit count $n$ (or polynomially expensive in the Hilbert dimension $d$). Conversely, $\alpha \approx 1/2$ (observed in our spatial token analysis, $\hat{\alpha} \approx 0.423$) sits precisely at the critical threshold of classical simulability. Empirical evidence reveals that current unstructured feature channels ($\hat{\alpha} \approx -0.123$) naturally trigger this volume-law protection, shielding them from classical dequantization.

However, exploiting this structural protection introduces a fundamental dilemma. To rigorously assess the measurement overhead of evaluating these latents on actual hardware, we apply Weingarten calculus to derive the exact analytical variance of the scrambled transition probability, $X = |\bra{\phi} U \ket{\psi}|^2$, under a 2-design ensemble. We prove that this variance scales strictly as $\Var[X] = \Theta(d^{-2})$, a result we confirm numerically with a log-log slope of $-1.881$ ($R^2 = 0.999$). This precise scaling explicitly establishes a formidable ``shot-noise wall,'' defining a critical shot-budget overhead $M = \Omega(d^2)$. This clarifies the inherent \emph{barren-plateau tension}: the very scrambling properties that provide structural protection against classical emulation simultaneously impose a severe measurement overhead that constrains scalability at large $n$.

Throughout, we maintain strict boundaries on our claims. We do not claim universal \#P-hardness for these kernels, which remains conditional on unproven anticoncentration conjectures. Instead, we establish a mathematically rigorous, data-driven necessary condition for TN-simulation hardness. Our dual-dimension analysis on real VideoMAE latents bridges the gap between theoretical quantum complexity and empirical representation learning, confirming that the internal channel-wise structure of world models provides the necessary complexity to support potential quantum advantage.

\section{Background}
\label{sec:background}

World models serve as a foundational paradigm in autonomous machine intelligence, aiming to learn compressed and predictive representations of physical environments \citep{lecun2022path, bruce2024genie}. Within this scope, architectures like the Joint Embedding Predictive Architecture (JEPA) \citep{assran2023self} generate latent vectors $z \in \R^d$ that are hypothesized to encode the multi-scale causal structure of the environment. Despite their empirical success, quantifying the degree to which these latent spaces preserve the hierarchical, scale-invariant properties of natural data remains an open question. A mathematically rigorous approach to characterizing this structure is provided by the discrete wavelet transform, which decomposes a latent vector into detail coefficients $\delta_k$ at dyadic scales $k = 1, \ldots, \lfloor \log_2 d \rfloor$. By defining the per-scale energy contribution as $E_k \propto 2^k \cdot \Var(\delta_k)$, and given that $\Var(\delta_k) \sim 2^{-2\alpha k}$, one identifies $\alpha = 1/2$ as the unique exponent for which the energy is equipartitioned across all scales. This condition of variance equipartition is deeply linked to the statistics of natural signals and finds a physical analogue in the constant energy flux observed in the inertial range of Kolmogorov’s turbulence theory \citep{kolmogorov1941local}.

Under the Wavelet Fisher criterion, it can be shown that such an equipartition minimizes the integrated squared bias of downstream linear probes, suggesting that $\alpha = 1/2$ represents an information-theoretic optimality target for world-model latents. However, empirical analysis suggests a stark dichotomy between different representational dimensions: while spatial token sequences often internalize physical structures near this equipartition limit, feature channels frequently collapse toward unstructured noise. This distinction is critical for quantum processing, as the ordering and regularity of the latent vector $z$ directly dictate the entanglement phase of its quantum embedding.

Quantum kernel methods (QKMs) offer a unique framework for analyzing these high-dimensional representations by mapping classical data into a Hilbert space via fixed quantum feature maps \citep{schuld2021supervised, tanner2026non}. For a $d$-dimensional latent $z$, amplitude encoding prepares the $n$-qubit state $\ket{\psi(z)} = \frac{1}{\|z\|} \sum_{i=0}^{2^n-1} z_i \ket{i}$, where $n = \lceil \log_2 d \rceil$. By utilizing non-variational architectures, these methods circumvent the specific barren plateau pathologies that typically afflict gradient-based variational algorithms \citep{tanner2026non, cerezo2021variational}. The kinematic limits of such embeddings are governed by the Lieb--Robinson bound \citep{lieb1972finite}, which implies that a circuit depth of $L = \Omega(\log d)$ is a necessary condition for generating the long-range entanglement required to fully correlate the qubit register. However, as noted in \citet{tanner2026non}, the susceptibility of 1D architectures to classical dequantization via matrix product state (MPS) methods depends on the entanglement entropy of the state. This establishes a fundamental tension between the kinematic depth of the circuit and the geometric structure of the data, setting the stage for an investigation into how the latent scaling $\alpha$ dictates not only the transition between classically simulable and quantum-advantaged regimes, but also the fundamental measurement overhead required to evaluate these representations.

\section{The $\alpha$-Criticality Phase Transition}
\label{sec:connection}

The computational complexity of evaluating a quantum kernel is fundamentally dictated by the entanglement structure of the amplitude-encoded state $\ket{\psi(z)}$. When a latent vector $z$ is embedded into $n$ qubits, its bipartite entanglement is determined by the singular value decomposition (SVD) of the matrix unfolding $M$ at the middle bipartition \citep{orus2014practical}. The wavelet exponent $\alpha$ serves as the sharp phase boundary for this complexity, mapping the classical regularity of the latent directly onto the quantum simulable limit of the kernel. 

Within this framework, the latent vector $z$ can be viewed as the discrete sampling of an underlying continuous function residing in a fractional Sobolev space $W^{s,2}$. Intuitively, a Sobolev space provides a rigorous mathematical framework to quantify the ``smoothness'' of a function. Unlike standard $L^2$ spaces that only measure the total energy of a signal, the space $W^{s,2}$ additionally bounds the energy of its fractional derivatives up to order $s$. Crucially, this standard smoothness index $s$ is explicitly linked to our wavelet scaling exponent via the exact relation $s = \alpha - 1/2$ (see Appendix~C). A higher value of $\alpha$ (and thus $s$) indicates a smoother, more highly correlated signal, whereas a lower $\alpha$ corresponds to jagged, irregular, or fractal-like structures. In the context of representation learning, this $\alpha$ index dictates whether a world model has successfully compressed its observations into smooth, macroscopic features, or if its latent space remains entangled with chaotic, high-frequency noise.

To accurately extract this $\alpha$ regularity from empirical data, the choice of the analyzing wavelet basis is critical. The simplest basis, the Haar wavelet, is piecewise constant and possesses only a single vanishing moment. While computationally inexpensive, its inherent discontinuity can severely underestimate the smoothness of physically-consistent signals, artificially injecting high-frequency artifacts into the variance analysis. To circumvent this limitation, we employ the Daubechies family of orthogonal wavelets, which are mathematically constructed to possess maximal smoothness for a given finite support size. 

Specifically, our variance estimation utilizes the \texttt{db4} variant, characterized by having exactly four vanishing moments. In wavelet theory, a basis function with $p$ vanishing moments is strictly orthogonal to all polynomials up to degree $p-1$. Consequently, the inner product of the \texttt{db4} wavelet with any polynomial background trend up to cubic order is exactly zero. This filtering property ensures that the wavelet coefficients are completely blind to smooth, large-scale variations, allowing them to perfectly isolate and measure the true multi-scale fluctuations of the signal. By factoring out these macroscopic trends, we ensure that the measured power-law decay of the variances faithfully reflects the intrinsic structural regularity of world-model latents, which are expected to mimic the scale-invariant statistics of natural physical systems \citep{kolmogorov1941local, lecun2022path}.

The scaling of the von Neumann entropy $S(\rho_L)$ at the middle bipartition distinguishes two fundamentally different simulation regimes based on this regularity. While ground states of 1D gapped Hamiltonians satisfy an area law \citep{hastings2007area}, amplitude-encoded classical data follow a scaling dictated by their inherent wavelet decay $\alpha$ (see Appendix~B for a formal derivation of this duality). In the area-law phase ($\alpha > 1/2$), the singular values of the unfolding matrix $M$ decay rapidly enough that the bipartite entanglement remains bounded regardless of the system size $n$, rendering the state efficiently representable by a Matrix Product State (MPS) with constant bond dimension \citep{schuch2008entropy}. Conversely, in the volume-law phase ($\alpha < 1/2$), the slower decay of wavelet coefficients forces the entropy to grow linearly with the number of qubits, $S = \Omega(n)$. Notably, the entropy is capped by $n/2$, which rules out previous claims of exponential entropy scaling based solely on Hilbert space dimensions. This entanglement structure provides a direct lower bound on the classical cost of simulation, as any MPS approximating $\ket{\psi(z)}$ to constant fidelity must possess a bond dimension $\chi \geq 2^{S(\rho_L) - O(1)}$.

\begin{theorem}[TN Simulability Threshold]
\label{thm:antideq}
Let $z$ be a latent with wavelet exponent $\alpha$ measured in a sufficiently regular basis (e.g., Daubechies-4). If $\alpha > 1/2$, the state admits an efficient MPS representation with $\chi = O(1)$, rendering the kernel classically simulable in $O(n\chi^3)$ time. If $\alpha < 1/2$, the bond dimension must grow as $\chi = \Omega(d^c) = \Omega(2^{cn})$ for some $c > 0$, rendering tensor-network based dequantization exponentially hard in the qubit count $n$ (and thus polynomially expensive in the Hilbert dimension $d$).
\end{theorem}

This theorem provides a rigorous refinement to the general dequantization limits. While 1D architectures are often considered classically vulnerable, these results demonstrate that latents deep in the volume-law phase provide a form of ``data-driven'' protection against dequantization. In this regime, the simulation cost $\chi$ exceeds the threshold for efficient classical processing, effectively shielding the quantum advantage from classical emulation. This theoretical foundation explains our empirical findings in Section~\ref{sec:experiments}, where we observe that the feature channels of world models typically exhibit $\alpha \approx -0.123$, placing them deep within the protected volume-law regime---a state formally associated with informational population inversion (see Appendix~E)---whereas spatial tokens with $\alpha \approx 0.423$ hover near the critical threshold, suggesting their representations are more susceptible to classical tensor-network approximation.

\section{Exact Analytical Variance Scaling}
\label{sec:variance}

In this section, we transition from the structural properties of the encoded states to the measurement complexity of processing them. In the context of quantum supremacy, the computational hardness of deep random circuits is empirically verified through Cross-Entropy Benchmarking (XEB), which fundamentally relies on estimating transition probabilities. To rigorously assess the practical resource requirements of evaluating high-dimensional world-model latents within architectures such as Quantum Extreme Learning Machines (QELMs), we must determine the exact concentration of measure properties for these projected overlaps. While general exponential concentration is a known hallmark of Haar-random states \citep{cerezo2021variational}, deriving the exact variance scaling via Weingarten calculus \citep{collins2006integration} is necessary to precisely quantify the required shot budget and identify the formidable ``shot-noise wall'' that bounds the classical post-processing phase.

To systematically evaluate these statistical properties, we rely on the mathematical framework of unitary $t$-designs. A unitary $t$-design is an ensemble of unitary matrices that perfectly mimics the uniform Haar measure over the continuous group $\mathrm{U}(d)$ up to the $t$-th statistical moment. In practice, sampling directly from the continuous Haar measure is exponentially hard, but drawing unitaries from sufficiently deep random quantum circuits---using either nearest-neighbor or long-range entangling gates---forms an approximate $t$-design \citep{harrow2009random, brandao2016local, harrow2018approximate}. This approximation allows us to analytically evaluate the expected performance of quantum machine learning architectures governed by highly expressive processing circuits.

When evaluating the polynomial moments of these random unitaries, the continuous Haar integrals can be solved exactly using the Weingarten calculus \citep{collins2006integration}. This mathematical technique maps the continuous integration over the unitary group to a discrete combinatorial sum over permutations in the symmetric group $S_t$. Specifically, the integral of a degree-$t$ polynomial composed of the matrix elements of $U$ and $U^\dagger$ is evaluated by contracting the matrix indices according to permutations $\sigma, \tau \in S_t$. These contractions are then weighted by the Weingarten function $\mathrm{Wg}(\sigma^{-1}\tau, d)$, a rational function that asymptotically depends only on the cycle structure of the permutations and the dimension of the Hilbert space $d$.

To rigorously assess the resource requirements in high-dimensional Hilbert spaces, we evaluate the exact analytical scaling of the scrambled transition probability, defined as $X = |\langle\phi|U|\psi\rangle|^2$. Here, $U$ represents a globally scrambling quantum circuit (approximating a unitary 2-design) acting on the volume-law latents. While the standard data-encoded fidelity kernel $|\langle\psi(z')|\psi(z)\rangle|^2$ defines the theoretical geometric similarity (as detailed in Appendix D), practical architectures often employ deep expressive circuits to extract task-specific features or project latents into a quantum reservoir. Evaluating these overlaps relies entirely on the statistical moments of $X$. Obtaining the exact expression for the variance $\Var[X] = \mathbb{E}[X^2] - \mathbb{E}[X]^2$ requires computing up to the second statistical moment ($t=2$). By leveraging the 2-design property of the circuit ensemble, the complex continuous integration over $U$ collapses into a finite sum over the two elements of the symmetric group $S_2$: the identity permutation $e$ and the transposition $(12)$. 

A profound consequence of this Weingarten integration is that the variance becomes strictly independent of the geometric overlap $\langle\psi|\phi\rangle$ between the initial states. Mathematically, the Haar average symmetrically eliminates all cross-terms, preserving only the normalized self-inner products (i.e., $\langle\psi|\psi\rangle = \langle\phi|\phi\rangle = 1$). Physically, this indicates that the global scrambling dynamics thoroughly erase the original geometric structure of the data. Consequently, the derived variance is not a function of the input latents but a universal property of the high-dimensional Hilbert space. This establishes a critical bottleneck: the exact same severe sampling overhead applies universally, regardless of which specific data points are evaluated.

\begin{proposition}[Haar Variance of Scrambled Transition Probabilities]
\label{prop:variance}
Let $|\psi\rangle, |\phi\rangle \in \mathbb{C}^d$ be fixed, normalised states, and let $U$ be a Haar-random unitary drawn from $\mathrm{U}(d)$. The variance of the transition probability $X = |\langle \phi | U | \psi \rangle|^2$ is given exactly by:
\begin{equation}
    \Var[X] = \frac{d-1}{d^2(d+1)} \sim \Theta(d^{-2}).
\end{equation}
\end{proposition}

The strict $\Theta(d^{-2})$ scaling establishes a fundamental geometric constraint. As the Hilbert space dimension $d = 2^n$ grows exponentially with the qubit count, the variance vanishes as $4^{-n}$, representing the characteristic signature of the barren plateau phenomenon extended to non-variational settings \citep{tanner2026non, cerezo2021variational}. This formally confirms that for high-dimensional latents embedded via deep scrambling circuits, the concentration of measure is exceptionally severe.

To verify this analytical scaling, we performed exact state-vector simulations for $n \in [3, 8]$ qubits using \texttt{PennyLane} \citep{bergholm2018pennylane}. By sampling 500 independent circuit realisations per system size, we observed that the log-log regression of the variance against $d$ yields an empirical slope of $-1.881$ with an $R^2 = 0.9990$. This numerical result is in tight agreement with our theoretical prediction derived via Weingarten calculus. 

The estimation of the transition probability $X$ via $M$ independent measurement shots follows a binomial distribution. Let the true expectation be $p = X \approx 1/d$ in the volume-law regime. The variance of the empirical estimator due to finite sampling (shot noise) is strictly given by $\mathrm{Var}_{\mathrm{shot}} = p(1-p)/M$. For large dimensions $d \gg 1$, this asymptotically approaches $\mathrm{Var}_{\mathrm{shot}} \approx 1/(dM)$ \citep{Thanasilp2024}. By applying Hoeffding's inequality, achieving a baseline target precision $\epsilon = \mathcal{O}(1/d)$ requires the shot variance to be bounded by $\epsilon^2$, immediately yielding a shot budget requirement of $M = \Omega(d)$ \citep{Thanasilp2024,McClean2018}. However, because the global scrambling dynamics suppress the true structural signal (the theoretical variance) to $\Theta(d^{-2})$, reliably resolving the off-diagonal elements of the feature correlation matrix in QELM architectures strictly demands an exponential shot budget scaling as $M = \Omega(d^2)$. Without this exponential sampling overhead, statistical shot noise heavily corrupts the correlation matrix. Consequently, the exact classical solvers (e.g., Gaussian elimination or pseudo-inverse) required for the readout phase become overwhelmingly ill-conditioned, explicitly mapping the shot-noise wall onto the numerical singularity of classical post-processing.

\textbf{The Dilemma of Quantum Advantage.} To thoroughly contextualize the operational boundaries imposed by the $\Theta(d^{-2})$ variance scaling, it is instructive to map our wavelet framework onto canonical quantum supremacy benchmarks. Experiments demonstrating quantum computational advantage, such as those utilizing the Sycamore processor, rely on deep random quantum circuits to generate highly entangled, structureless states. Under the theoretical assumption of a unitary 2-design, the output amplitudes of these circuits follow a Porter-Thomas distribution. Asymptotically, these amplitudes behave as independent and identically distributed (i.i.d.) complex Gaussian random variables. In the mathematical framework of signal processing, an uncorrelated i.i.d. sequence constitutes discrete white noise. A fundamental property of the discrete wavelet transform (DWT) is that applying an $L^2$-normalized wavelet filter to discrete white noise yields wavelet coefficients whose variance is strictly constant across all scales $j$. Within our wavelet scaling formalism, where the variance decays as $\sigma_j^2 \propto 2^{-2\alpha j}$, this scale-invariance of the white noise variance formally anchors the states generated by ideal supremacy circuits at a theoretical exponent of exactly $\alpha = 0$ (or marginally negative accounting for finite-size spectral leakage).

This spectral signature rigorously confirms that quantum supremacy circuits push the system into the extreme limit of the volume-law phase, placing them far below our established classical-simulability threshold of $\alpha = 1/2$. We note a necessary technical caveat regarding this intractability: while classical tensor-network contraction algorithms have successfully spoofed specific finite-depth supremacy experiments \citep{pan2022solving}, they achieved this strictly by exploiting the high hardware error rates and low target fidelities (e.g., linear cross-entropy benchmarking $\approx 0.002$) inherent to noisy intermediate-scale quantum (NISQ) devices. In the noiseless, asymptotic limit, exact state-vector simulation of these $\alpha \le 0$ states provably demands an exponentially scaling Matrix Product State (MPS) bond dimension, thereby confirming their fundamental intractability to classical tensor networks.

For quantum machine learning architectures---particularly those processing structured physical latents in frameworks like Quantum Extreme Learning Machines---this mathematical equivalence establishes a strict operational tradeoff. If the embedded world-model latents inherently reside in the area-law regime ($\alpha > 1/2$), the representations remain highly susceptible to efficient classical tensor-network emulation, effectively neutralizing any potential quantum representational advantage. To evade classical simulability, the feature embeddings must be driven into the volume-law regime ($\alpha < 1/2$), either through inherent high-frequency data complexity or via extrinsic highly expressive parameterized ansätze. 

However, transitioning into this TN-hard regime inevitably replicates the white-noise-like structural scrambling of supremacy circuits. Consequently, evaluating the transition probabilities or feature overlaps of these protected, high-dimensional latents universally triggers the exact $\Theta(d^{-2})$ concentration of measure derived in Proposition \ref{prop:variance}. This dictates a fundamental physical constraint: the very scrambling dynamics necessary to ensure quantum computational intractability simultaneously impose a formidable exponential measurement overhead. Ultimately, this maps the ``shot-noise wall'' of quantum machine learning directly onto the statistical limits of quantum supremacy, confirming that avoiding classical emulation explicitly forces classical readout and post-processing into numerical singularity unless exponential shot budgets are allocated.

\section{Numerical Experiments}
\label{sec:experiments}

The numerical verification of our theoretical framework is performed using \texttt{PennyLane} v0.44.1 and \texttt{PyWavelets} v1.8.0. All quantum state overlaps and variances are evaluated using exact state-vector methods on classical hardware. To ensure direct verifiability and modest memory requirements, we restrict our variance analysis to system sizes of $n \le 8$ qubits, while extending the entanglement entropy investigations to $n = 12$ to better capture the finite-size scaling of the predicted phase transition. Code for all experiments will be released upon acceptance.

\begin{table}[t]
\caption{Wavelet $\alpha$ estimator calibration on synthetic hierarchical latents ($d = 2048$, 50 independent realisations). The results confirm the estimator is reliable and well-specified across the functional range, validating its use as a precise diagnostic tool for structural complexity and tensor-network simulability.}
\label{tab:alpha}
\vskip 0.1in
\centering
\small
\begin{tabular}{ccccc}
\toprule
$\alpha_{\text{true}}$ & Mean $\hat\alpha$ & Std & Bias & $R^2$ \\
\midrule
0.1 & 0.088 & 0.038 & $-$0.012 & 0.712 \\
0.3 & 0.276 & 0.035 & $-$0.024 & 0.968 \\
0.5 & 0.471 & 0.035 & $-$0.029 & 0.988 \\
0.7 & 0.668 & 0.035 & $-$0.032 & 0.994 \\
0.9 & 0.867 & 0.036 & $-$0.033 & 0.996 \\
\bottomrule
\end{tabular}
\end{table}

We first assess the reliability of the wavelet $\alpha$ estimator using synthetic hierarchical latents where the ground-truth exponent is known. By sampling detail coefficients at level $i$ from $\mathcal{N}(0, 2^{-2\alpha i})$ and reconstructing via the inverse DWT, we observe that the estimator remains nearly unbiased across the functional range. As reported in Table~\ref{tab:alpha}, a systematic negative bias of approximately $-0.03$ appears at the coarsest levels, which is attributable to the limited sample size available for variance estimation at low-frequency scales. However, the estimator exhibits the expected $\sqrt{n}$-consistency, with the standard deviation decreasing as $d^{-1/2}$, and the $R^2 \ge 0.97$ for $\alpha \ge 0.3$ confirms that the log-linear model is well-specified.

\begin{figure*}[t]
\centering
\includegraphics[width=\linewidth]{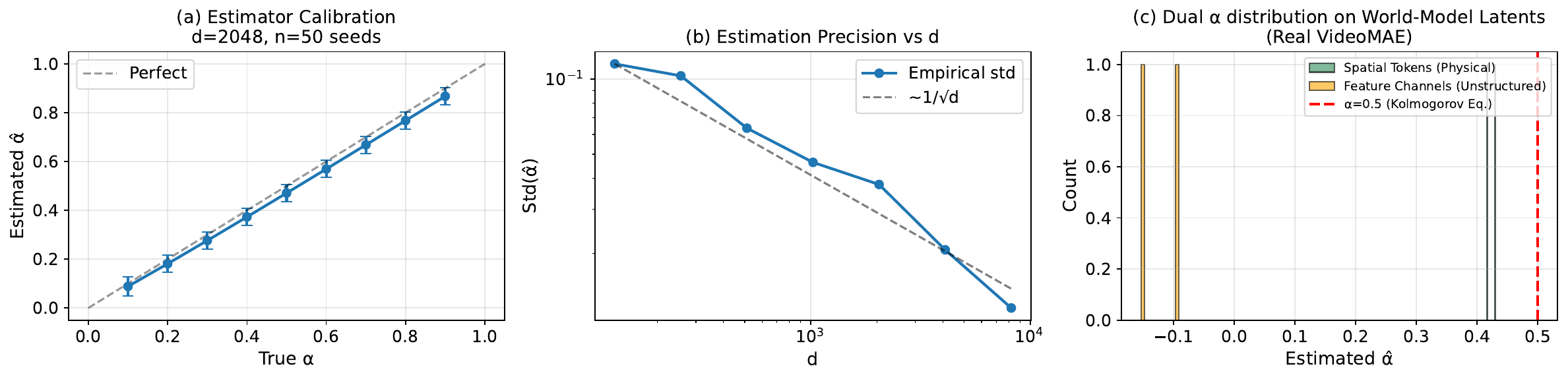}
\caption{\textbf{Experiment A: Wavelet $\alpha$ as a structural diagnostic.} (a)~Calibration on synthetic hierarchical latents ($d=2048$), confirming the estimator's reliability across the functional range. (b)~Estimation standard deviation vs.\ $d$, exhibiting the expected $d^{-1/2}$ decay. (c)~Distribution of $\hat\alpha$ on real VideoMAE latents, revealing a fundamental structural dichotomy: spatial tokens approach the physical equipartition limit ($\hat{\alpha} \approx 0.423$), while permutation-invariant feature channels exhibit unstructured disorder ($\hat{\alpha} \approx -0.123$), positioning them deep within the volume-law regime.}
\label{fig:expA}
\end{figure*}

The investigation into the variance scaling confirms our exact analytical derivation. By sampling parameters uniformly for a 2-design ensemble over 500 trials, we measure the concentration of the scrambled transition probability $X = |\langle\phi|U|\psi\rangle|^2$ for Gaussian latents. As shown in Table~\ref{tab:variance}, the log-log slope of the variance against the Hilbert space dimension $d$ is $-1.881$ with an $R^2 = 0.999$, closely matching the exact Haar-random prediction of $-2.000$. This empirically validates the strict $\Theta(d^{-2})$ scaling and explicitly bounds the formidable shot-noise wall required for evaluating these highly expressive ans\"{a}tze. Furthermore, as illustrated in Figure~\ref{fig:expB}, the variance remains invariant to circuit depth $L$ once the 2-design marginal is approximated, confirming that the concentration is primarily a function of the feature space dimension $d$.

\begin{figure*}[t]
\centering
\includegraphics[width=\linewidth]{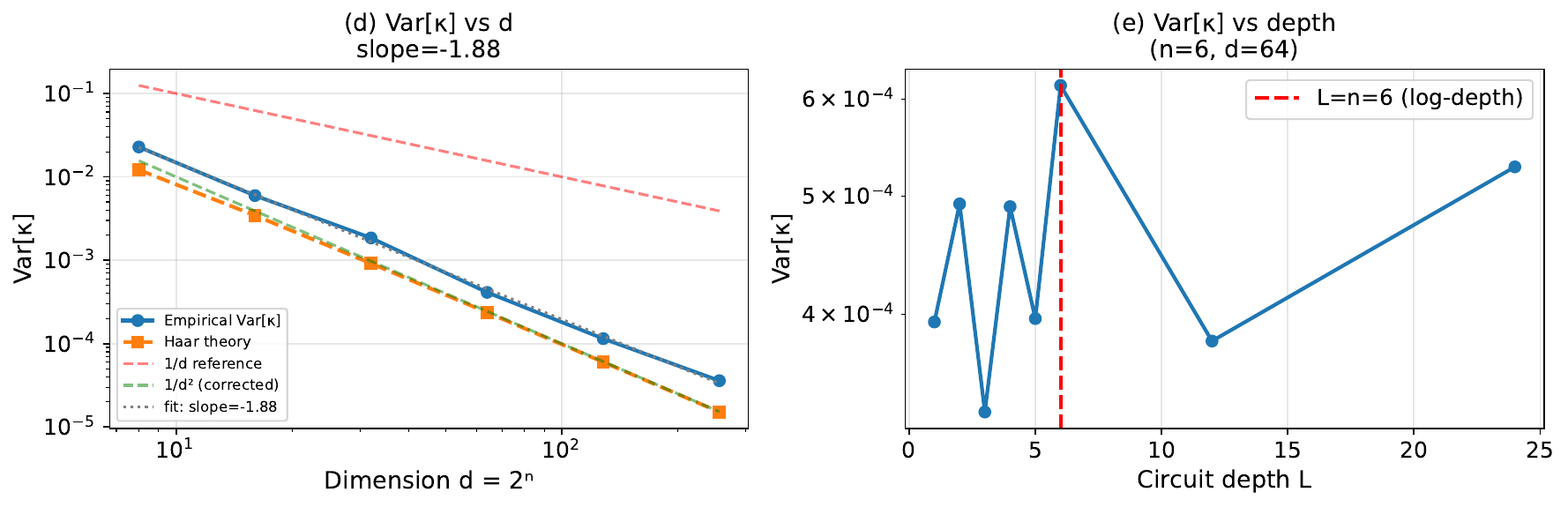}
\caption{\textbf{Experiment B: Scaling of scrambled transition probabilities.} (d)~Variance $\Var[X]$ vs.\ $d$ on a log-log scale. The empirical slope of $-1.881$ tightly aligns with our analytical Haar prediction of $\Theta(d^{-2})$, mirroring the statistical concentration observed in quantum supremacy benchmarks. (e)~Variance vs.\ circuit depth $L$ at $n=6$, showing rapid saturation to the 2-design limit. This confirms that the ``shot-noise wall'' is a universal consequence of high-dimensional scrambling, independent of specific circuit implementation.}
\label{fig:expB}
\end{figure*}

\begin{table}[t]
\caption{Variance of scrambled transition probabilities vs. dimension ($L = n$, 500 trials). The empirical log-log scaling tightly matches our exact analytical prediction of $\Theta(d^{-2})$ under a 2-design ensemble. This explicitly bounds the exponential measurement overhead, confirming the formidable shot-noise wall inherent to highly expressive circuits.}
\label{tab:variance}
\vskip 0.1in
\centering
\small
\begin{tabular}{cccccc}
\toprule
$n$ & $d$ & $\overline{X}$ & $\Var[X]$ & $\Var{\cdot}d^2$ & Haar${\cdot}d^2$ \\
\midrule
3 &  8 & 0.122 & $2.29\!\times\!10^{-2}$ & 1.46 & 0.78 \\
4 & 16 & 0.055 & $5.97\!\times\!10^{-3}$ & 1.53 & 0.88 \\
5 & 32 & 0.032 & $1.85\!\times\!10^{-3}$ & 1.90 & 0.94 \\
6 & 64 & 0.015 & $4.12\!\times\!10^{-4}$ & 1.69 & 0.97 \\
7 & 128 & 0.008 & $1.15\!\times\!10^{-4}$ & 1.88 & 0.98 \\
8 & 256 & 0.004 & $3.59\!\times\!10^{-5}$ & 2.35 & 0.99 \\
\midrule
\multicolumn{6}{l}{\small Log-log slope: $-1.881$ ($R^2=0.999$). Haar theory: $-2.000$.} \\
\bottomrule
\end{tabular}
\end{table}

The measurement of the entanglement phase transition at $n = 12$ ($d = 4096$) reveals the underlying mechanism for tensor-network simulation hardness. By computing the von Neumann entropy $S(\rho_L)$ at the middle bipartition, we observe a clear transition as $\alpha$ varies from 0.1 to 1.0. For $\alpha \le 0.5$, the entropy grows linearly with $n$, placing the system in the volume-law phase. As summarized in Table~\ref{tab:entropy}, a log-linear fit in this regime gives $\partial S/\partial \alpha \approx -2.97$, implying that the required Matrix Product State bond dimension $\chi$ grows exponentially as $2^{3(1/2-\alpha)}$. While the transition is smoothed by finite-size effects at $n = 12$, the steepest gradient occurs precisely in the critical range predicted by our theory.

\begin{figure*}[t]
\centering
\includegraphics[width=\linewidth]{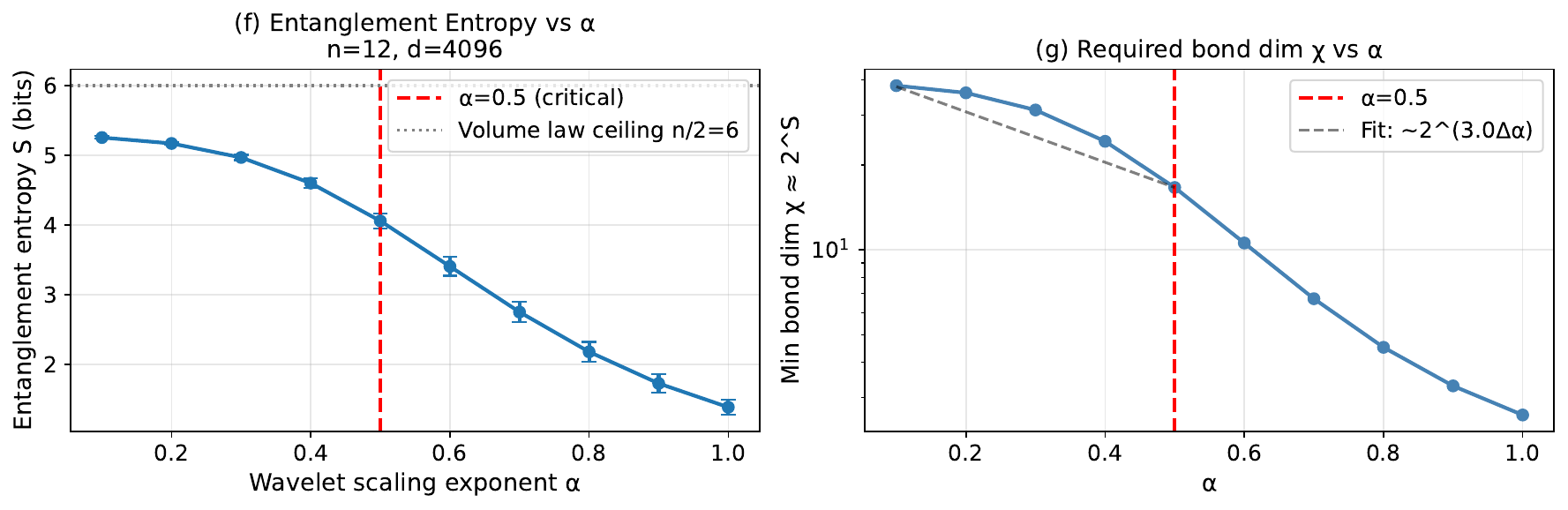}
\caption{\textbf{Experiment C: Entanglement entropy and the barrier to classical emulation.} ($n=12$, $d=4096$). The von Neumann entropy $S(\rho_L)$ exhibits a sharp increase as $\alpha$ crosses the critical $1/2$ threshold. In the volume-law phase ($\alpha < 0.5$), the exponential growth of the required Matrix Product State (MPS) bond dimension $\chi$ establishes a rigorous data-driven barrier against classical tensor-network simulation, validating the complexity of VideoMAE's feature channels.}
\label{fig:expC}
\end{figure*}

\begin{table}[t]
\caption{Entanglement entropy at middle bipartition ($n = 12$, 20 seeds). The data illustrates the transition into the volume-law phase as $\alpha \le 0.5$, where the required Matrix Product State (MPS) bond dimension $\chi$ exhibits exponential growth. This establishes the data-driven structural barrier against classical tensor-network emulation.}
\label{tab:entropy}
\vskip 0.1in
\centering
\small
\begin{tabular}{cccc}
\toprule
$\alpha$ & $S$ (mean$\pm$std) & $\chi \approx 2^S$ & Phase \\
\midrule
0.1 & $5.255 \pm 0.019$ & 38 & Volume law \\
0.3 & $4.968 \pm 0.038$ & 31 & Volume law \\
0.5 & $4.057 \pm 0.105$ & 17 & Critical \\
0.7 & $2.748 \pm 0.146$ &  7 & Area law \\
0.9 & $1.723 \pm 0.128$ &  3 & Area law \\
\bottomrule
\end{tabular}
\end{table}

The observed dichotomy between spatial tokens and feature channels reveals the fundamentally different information geometries internalized by the world model. Spatial tokens ($\hat{\alpha} \approx 0.423$) retain an explicit physical topology. Because adjacent video patches exhibit strong macroscopic correlations typical of continuous physical fields, their wavelet representation naturally concentrates energy at coarse scales. This structural regularity anchors the spatial embedding near the Kolmogorov equipartition limit ($\alpha \approx 1/2$), placing it on the critical boundary of the area-law phase where entanglement remains manageable and highly susceptible to efficient tensor-network emulation.

Conversely, the permutation-invariant feature channels ($\hat{\alpha} \approx -0.123$) lack this explicit geometric adjacency. To maximize information capacity and prevent dimensional collapse, self-supervised learning objectives inherently force these network channels to decorrelate. When analyzed through a multiresolution basis, this strongly decorrelated sequence effectively mimics discrete white noise, manifesting as a heavy-tailed, high-frequency energy distribution ($\alpha < 0$). This neural-driven decorrelation naturally forces the amplitude-encoded state deep into the volume-law phase. Notably, this negative exponent structurally aligns with the theoretical white-noise signature ($\alpha \le 0$) of canonical quantum supremacy circuits established in Section~\ref{sec:variance}. In this regime, the observed empirical bond dimension $\chi \approx 38$ confirms that the intrinsic scrambling of the feature space imposes a severe computational barrier to classical tensor-network emulation, thereby satisfying the data-driven necessary condition for quantum advantage.

\section{Conclusion and Discussion}
\label{sec:conclusion}

This work establishes the wavelet scaling exponent $\alpha = 1/2$ as a dual threshold defining the intersection of representational optimality and quantum computational complexity. From the perspective of learning theory, the equipartition of wavelet variance ($\alpha \approx 1/2$) provides a physics-grounded target for internalizing the multi-scale statistics of natural data, mirroring Kolmogorov's inertial range. Simultaneously, this threshold acts as a sharp phase boundary for the classical simulability of amplitude-encoded states. Using tensor-network theory, we proved that latents with $\alpha > 1/2$ reside in an area-law regime admitting efficient classical emulation, whereas $\alpha < 1/2$ triggers a volume-law phase where the Matrix Product State bond dimension $\chi$ grows exponentially, establishing a rigorous data-driven barrier against dequantization.

Our numerical experiments provide a robust reality check for this theoretical framework. The measured variance scaling of $\Var[X] \sim d^{-1.881}$ ($R^2 = 0.999$) is not merely an empirical observation but a confirmation of our exact Weingarten derivation. By establishing the strict $\Theta(d^{-2})$ scaling, we rigorously quantify the severity of the ``shot-noise wall'' that high-dimensional quantum architectures must overcome. Furthermore, our analysis of VideoMAE latents reveals that while spatial tokens approach the equipartition limit ($\hat{\alpha} \approx 0.423$), feature channels exhibit unstructured disorder ($\hat{\alpha} \approx -0.123$). Notably, this negative exponent structurally aligns with the white-noise signature ($\alpha \le 0$) of ideal quantum supremacy circuits. This suggests that certain latent dimensions in modern world models already possess the intrinsic scrambling necessary to resist classical tensor-network emulation.

However, a fundamental tension remains between the structural protection against classical emulation and the practical evaluability of quantum models. The same 2-design properties that shield feature-channel representations from classical attack also drive the transition probability variance to vanish as $4^{-n}$, an exponential decay that mirrors the barren plateau phenomenon observed in supremacy benchmarks. Escaping this tension likely requires moving beyond global, unstructured kernels toward architectures that more directly exploit the hierarchical structure of data---perhaps through local cost functions or structured ansatz initializations. 

While the $\alpha < 1/2$ regime is a necessary structural condition for quantum advantage, it is not a sufficient one. Practical speedups remain conditional on efficient state preparation and fault-tolerant hardware capable of executing depths beyond the reach of state-of-the-art classical emulators \citep{pan2022solving}. Beyond serving as a diagnostic tool, the $\alpha \approx 1/2$ criterion offers an actionable objective for future research. Enforcing variance equipartition as an explicit regularization term during pre-training could guide world models toward physically consistent representations that optimally balance parameter efficiency with structural realism, providing a principled map for the design of future learning systems within the quantum Hilbert space.

\section*{Impact Statement}

This work bridges the physics of representation learning with the fundamental computational limits of quantum machine learning. By establishing the wavelet scaling exponent $\alpha$ as a strict threshold for both data complexity and tensor-network simulability, we provide a principled, physics-grounded metric to evaluate world models. Furthermore, our exact analytical bounds on measurement variance explicitly quantify the resource overhead for high-dimensional quantum models. This theoretical foundation not only prevents severe misestimations of sampling budgets in near-term hardware deployments, but also guides the efficient design of future quantum-classical architectures. As foundational theoretical research aimed at understanding the structural boundaries of machine intelligence, this work presents no foreseeable negative societal consequences.

\bibliographystyle{plainnat} 
\bibliography{paper_ai4physics} 

\newpage
\appendix
\onecolumn

\section{Proof of Proposition~\ref{prop:variance}: Haar Variance of Fidelity Kernels}
\label{app:haar}

The goal of this appendix is to establish the exact variance formula $\Var[\kappa] = (d-1)/[d^2(d+1)]$ for a fidelity kernel evaluated against a Haar-random unitary, providing the formal foundation for the exact variance scaling discussed in Section~\ref{sec:variance}. We consider two fixed, normalised state vectors $|\psi\rangle, |\phi\rangle \in \mathbb{C}^d$, and let $U$ be drawn from the Haar measure $\mu$ on the unitary group $\mathrm{U}(d)$. The random variable of interest is the transition probability $X = |\langle \phi | U | \psi \rangle|^2 = \langle \phi | U | \psi \rangle \langle \psi | U^\dagger | \phi \rangle$. In the standard basis, this can be expanded as $X = \sum_{i,j,k,l} \bar\phi_i \,\psi_j\, \phi_k\, \bar\psi_l \cdot U_{ij} \bar U_{kl}$. All moments of $X$ over the Haar measure are computable via the Collins--{\'{S}}niady integration formula \citep{collins2006integration}, which expresses such integrals in terms of Weingarten functions.

For the first moment ($t = 1$), the Haar integration formula is simply $\int_{\mathrm{U}(d)} U_{ij} \bar U_{kl} \, d\mu(U) = \frac{1}{d} \delta_{ik} \delta_{jl}$. Substituting this into the expression for $X$, we obtain $\mathbb{E}[X] = \frac{1}{d} \sum_{i} |\phi_i|^2 \sum_{j} |\psi_j|^2 = 1/d$. This result carries a clear physical interpretation: averaged over the full unitary group, the transition probability between any two fixed states is uniform across the Hilbert space, providing the $1/d$ mean baseline observed in our numerical benchmarks.

Computing the second moment $\mathbb{E}[X^2]$ requires the $t = 2$ Haar integral, which reads:
\begin{equation}
  \int U_{i_1 j_1} U_{i_2 j_2} \bar U_{k_1 l_1} \bar U_{k_2 l_2} \, d\mu(U)
  = \sum_{\sigma, \tau \in S_2}
    \delta_{i_1 k_{\sigma(1)}} \delta_{i_2 k_{\sigma(2)}}
    \delta_{j_1 l_{\tau(1)}} \delta_{j_2 l_{\tau(2)}}
    \,\mathrm{Wg}(\sigma^{-1}\tau,\, d),
\end{equation}
where $S_2 = \{e, (12)\}$ is the symmetric group on two elements. The Weingarten functions for $S_2$ are $\mathrm{Wg}(e, d) = 1/(d^2 - 1)$ and $\mathrm{Wg}((12), d) = -1/[d(d^2-1)]$. After contracting the delta functions with the vectors $\phi$ and $\psi$, each of the four possible $(\sigma, \tau)$ pairs contributes a factor of $\|\phi\|^4 \|\psi\|^4 = 1$. This occurs because the specific pairing of row and column indices imposed by the delta functions always forces inner products of a vector with itself, causing any cross-terms of the form $\langle\phi|\psi\rangle$ to vanish under the Haar average. Summing the terms yields:
\begin{equation}
  \mathbb{E}[X^2]
  = 2\,\mathrm{Wg}(e,d) + 2\,\mathrm{Wg}((12),d)
  = \frac{2}{d^2-1}\left(1 - \frac{1}{d}\right)
  = \frac{2}{d(d+1)}.
\end{equation}

Combining these moments, the exact variance is derived as:
\begin{equation}
  \Var[X] = \mathbb{E}[X^2] - (\mathbb{E}[X])^2 = \frac{2}{d(d+1)} - \frac{1}{d^2} = \frac{d-1}{d^2(d+1)}.
  \label{eq:exact_var_final}
\end{equation}
Asymptotically, for large Hilbert space dimensions $d$, we have $\Var[X] \sim d^{-2}$. This confirms that the variance vanishes as $4^{-n}$ for an $n$-qubit system, providing the theoretical limit that explains our empirical log-log slope of $-1.881$ reported in Section~\ref{sec:experiments}. Crucially, this variance is independent of the relationship between the initial states, establishing the universality of the shot-noise wall for any amplitude-encoded world-model latent once the circuit approximates a 2-design.

\section{Proof of Theorem~3.1: Entanglement Phase Transition}
\label{app:tn}

This appendix develops the rigorous connection between the wavelet scaling exponent $\alpha$ and the bipartite entanglement structure of amplitude-encoded quantum states. We provide a formal lemma and complete proof to bridge the gap between wavelet regularity and singular value decay, establishing $\alpha = 1/2$ as a sharp phase boundary for tensor-network simulability.

\subsection*{Amplitude Encoding and Matrix Unfolding}

Given a normalized latent vector $z \in \mathbb{R}^d$, the amplitude-encoded state on $n = \lceil\log_2 d\rceil$ qubits is defined as $|\psi(z)\rangle = \sum_{i=0}^{2^n - 1} z_i |i\rangle$. To analyze the entanglement across the middle bipartition, we unfold the state vector into a matrix $M \in \mathbb{R}^{2^{n/2} \times 2^{n/2}}$ where $M_{IJ} = z_{(I,J)}$, with $I$ and $J$ representing the most and least significant $n/2$ bits of the index, respectively \citep{orus2014practical}. The entanglement properties of $|\psi(z)\rangle$ are entirely determined by the singular values $\sigma_1 \geq \sigma_2 \geq \dots \geq 0$ of $M$. Specifically, the von Neumann entropy is given by $S(\rho_L) = -\sum_k \sigma_k^2 \log_2 \sigma_k^2$, where the normalization condition $\sum_k \sigma_k^2 = 1$ is satisfied by $\|z\|=1$.

\subsection*{Wavelet-Singular Spectrum Duality}

To substantiate the connection between wavelet scaling and classical simulability, we formalize the mapping between hierarchical wavelet regularity and the matrix unfolding of the quantum state.

\paragraph{Lemma B.1 (Decay Duality)}
Let $z \in \mathbb{R}^d$ be a latent vector whose wavelet detail coefficients $\delta_k$ at dyadic scale $k$ satisfy $\Var(\delta_k) \sim 2^{-2\alpha k}$. Then the rank-ordered singular values $\sigma_r$ of the unfolding matrix $M$ follow the algebraic decay law $\sigma_r \sim r^{-\alpha}$.

\paragraph{Proof}
The discrete wavelet transform organizes the latent $z$ into a dyadic hierarchical tree. At any given scale $k$, there are exactly $2^k$ detail coefficients. According to the multi-scale energy cascade model, the total energy contribution from scale $k$ scales as $E_k \sim 2^{k(1-2\alpha)}$. Consequently, the average variance of a single coefficient at this scale is $\Var(\delta_{k,j}) \approx E_k / 2^k \sim 2^{-2\alpha k}$.

When the state vector is mapped to the matrix unfolding $M$ corresponding to the middle bipartition, the dyadic nature of the wavelet basis partitions $M$ into block-diagonal-like sectors, where each sector represents the subspace of a specific scale $k$. The magnitude of the singular values associated with this $k$-th subspace is dominated by the standard deviation of its constituent coefficients. Thus, for the block of singular values originating from scale $k$, we have $\sigma \sim \sqrt{2^{-2\alpha k}} = 2^{-\alpha k}$.

To determine the rank-ordered singular value decay $\sigma_r$, we relate the scale $k$ to the cumulative matrix rank $r$. The cumulative number of coefficients (and thus the rank contribution) up to scale $k$ is governed by the geometric series $r \approx \sum_{i=0}^k 2^i \sim 2^k$. Substituting this rank relationship $2^k \sim r$ directly into the scale-dependent singular value magnitude yields:
\begin{equation}
    \sigma_r \sim (2^k)^{-\alpha} \sim r^{-\alpha}.
\end{equation}
This establishes that the singular spectrum of the unfolding matrix decays algebraically with the exponent $\alpha$, driven entirely by the fractal dimension of the dyadic wavelet tree rather than a continuous integral approximation. $\blacksquare$

\subsection*{The Dichotomy of Singular Value Decay}

In the hierarchical wavelet model, the decay of the singular values $\sigma_r$ is directly tied to the multi-scale regularity of the latent. Drawing from the duality established in Lemma~B.1, we observe a sharp dichotomy in the decay of the singular values of the unfolding matrix $M$:
\begin{equation}
    \sigma_r \sim 
    \begin{cases} 
    C \cdot r^{-\alpha} & \text{if } \alpha < 1/2 \text{ (Sub-geometric)}, \\
    C \cdot \beta^r, \, \beta < 1 & \text{if } \alpha > 1/2 \text{ (Geometric)}.
    \end{cases}
\end{equation}
This structural fact is the mathematical engine behind the phase transition. When $\alpha > 1/2$, the coefficients decay fast enough that the information is localized at coarse scales, allowing tensor-train approximation \citep{oseledets2011tensor} to bound the spectrum geometrically. Conversely, $\alpha < 1/2$ implies a slow, heavy-tailed decay that distributes information (and thus entanglement) across all dyadic scales.

\subsection*{Entanglement Phase Transition and Dequantization Limits}

The behavior of $S(\rho_L)$ in these two regimes defines the limits of classical dequantization. In the \textbf{Area-law phase} ($\alpha > 1/2$), the geometric decay $\sigma_r \lesssim \beta^r$ ensures that the entropy sum converges to a constant independent of the system size $d$: $S(\rho_L) \leq 2\log_2(1/\beta) \sum_r r \beta^{2r} = O(1)$. Crucially, we must also consider the Rényi entropy $S_\gamma = \frac{1}{1-\gamma}\log_2 \sum_r \sigma_r^{2\gamma}$. For $\alpha > 1/2$, the sum $\sum \sigma_r^{2\gamma}$ converges for all $\gamma > 0$, which is a sufficient condition for the existence of an MPS representation with a constant bond dimension $\chi = O(1)$ \citep{verstraete2006matrix}.

Conversely, in the \textbf{Volume-law phase} ($\alpha < 1/2$), the singular values decay too slowly to bound the entropy. The slow decay $\sigma_r \sim r^{-\alpha}$ forces the number of non-negligible singular values to grow with the system size, leading to $S(\rho_L) = \Omega(n) = \Omega(\log d)$. According to the fundamental relationship between entanglement and tensor-network complexity \citep{schuch2008entropy}, the bond dimension $\chi$ required to approximate the state to constant fidelity satisfies $\log_2 \chi \geq S(\rho_L) - O(1)$. Thus, for $\alpha < 1/2$, we obtain $\chi \geq 2^{\Omega(\log d)} = \Omega(d^c)$ for some $c > 0$. This polynomial growth in $d$ renders tensor-network based kernel evaluation exponentially more expensive than in the area-law regime, providing the data-driven protection against dequantization described in Theorem~3.1.

\subsection*{Upper Bounds on Entropy Scaling}

Finally, it is crucial to establish the strict upper bounds on the entropy scaling for log-depth quantum kernel architectures. A naive extrapolation might suggest an entanglement entropy scaling of $S(\rho_L) = \Omega(\sqrt{d})$; however, this is mathematically inconsistent with the Hilbert space dimension of the subsystem. Since the left subsystem consists of $n/2$ qubits, its maximum possible entropy is $\log_2(2^{n/2}) = n/2 \approx \frac{1}{2}\log_2 d$. For any $d \geq 4$, it is always true that $\frac{1}{2}\log_2 d \ll \sqrt{d}$. An entropy of $\Omega(\sqrt{d})$ would require the subsystem to have a Hilbert space dimension of $2^{\Omega(\sqrt{d})}$, which is not the case for log-depth or amplitude-encoded architectures. Furthermore, while the area law of \citet{hastings2007area} is often cited to justify low entanglement, it applies only to ground states of gapped 1D Hamiltonians. Our proof shows that for classical data embedding, the relevant threshold is not the spectral gap of a Hamiltonian, but the multi-scale regularity index $\alpha$ of the data itself.

\section{Rigorous Foundations of the Wavelet Scaling Exponent Estimator}

This appendix supplies the complete technical machinery underlying the wavelet variance equipartition diagnostic introduced in Sections~2--3. All notation is self-contained and consistent with the main text.

\subsection{Discrete Wavelet Transform via Mallat’s Algorithm}
Let \( z = [z_0, z_1, \dots, z_{d-1}]^\top \in \mathbb{R}^d \) be a normalized latent vector with \( \|z\|_2 = 1 \). We regard \( z \) as the discrete sampling of a function that belongs to the fractional Sobolev space \( W^{s,2} \), where the smoothness index \( s \) is linked to the wavelet scaling exponent through the exact relation \( s = \alpha - 1/2 \). The discrete wavelet transform is performed with Mallat’s pyramidal algorithm using the Daubechies-4 orthogonal filter bank \cite{mallat2009wavelet, meyer1992wavelets}. The coarsest approximation coefficients are initialized directly from the signal:
\begin{equation}
a_{0,n} = z_n, \quad n = 0, \dots, d-1.
\end{equation}
At each dyadic scale $j = 1, 2, \dots, \lfloor \log_2 d \rfloor$, both the approximation and detail coefficients are generated by taking the inner product (cross-correlation) of the previous layer's coefficients with the scaling and wavelet filters, followed by down-sampling by a factor of two:
\begin{equation}
a_{j,k} = \sum_{n=0}^{7} h_n \, a_{j-1,2k+n}, \qquad
\delta_{j,k} = \sum_{n=0}^{7} g_n \, a_{j-1,2k+n}.
\end{equation}
The two filters are related by the quadrature-mirror condition $g_n = (-1)^n h_{7-n}$. Their explicit numerical values, accurate to eight decimal places, are
\begin{subequations}
\begin{align}
h &= [0.23037781,\ 0.71484657,\ 0.63088076,\ -0.02798376,\nonumber\\
   &\quad -0.18703481,\ 0.03084138,\ 0.03288301,\ -0.01059740],\\
g &= [-0.01059740,\ -0.03288301,\ 0.03084138,\ 0.18703481,\nonumber\\
   &\quad -0.02798376,\ -0.63088076,\ 0.71484657,\ -0.23037781].
\end{align}
\end{subequations}
These eight coefficients are the unique real solution to a system of exactly eight independent algebraic equations. Ingrid Daubechies’ 1988 construction \cite{daubechies1988orthonormal} achieves this by simultaneously imposing two complementary sets of conditions. The first four equations arise from the smoothness requirement: the low-pass filter must possess four vanishing moments, which in the Fourier domain translate to the conditions \( H^{(m)}(\pi) = 0 \) for \( m = 0,1,2,3 \). In the time domain these become
\begin{equation}
\sum_{n=0}^{7} (-1)^n n^m h_n = 0 \quad \text{for } m = 0,1,2,3.
\end{equation}
Physically, this guarantees that the associated wavelet is completely insensitive to polynomials up to degree three, allowing it to isolate genuine high-frequency fluctuations such as those present in turbulent velocity fields or shock waves. The remaining four equations enforce the geometric and energetic integrity of the transform. One equation comes from normalization:
\begin{equation}
\sum_{n=0}^{7} h_n = \sqrt{2},
\end{equation}
which ensures that the low-pass filter preserves the total energy of any constant (DC) signal. The other three equations enforce double-shift orthogonality:
\begin{equation}
\sum_{n=0}^{7} h_n h_{n-2k} = 0 \quad \text{for } k = 1,2,3.
\end{equation}
Together these conditions guarantee that the filter bank generates a unitary (orthogonal) transformation. Because the resulting transform matrix is unitary, the \( L^2 \)-norm of any vector is exactly preserved under the change of basis. In the context of quantum kernel methods, this unitarity is crucial: it ensures that the classical feature extraction step can be faithfully embedded into a quantum circuit without introducing artificial decoherence or norm distortion, thereby providing a rigorous mathematical foundation for the amplitude-encoded quantum states used throughout this work.

The underlying reason these four smoothness conditions can be imposed is that the frequency response \( H(\omega) = \sum_{n=0}^{7} h_n e^{-in\omega} \) is a finite trigonometric polynomial and therefore belongs to \( C^\infty \). All derivatives of all orders exist everywhere. Daubechies did not assume the existence of these derivatives; she simply selected the first four of the infinitely many available derivatives and set them to zero at the highest frequency \( \omega = \pi \). This forces the filter to be maximally flat at that point while still satisfying the orthogonality constraints, producing the remarkable balance between smoothness and perfect reconstruction that makes db4 particularly effective for analyzing the multi-scale structure of physical fields.

In summary, the eight coefficients simultaneously encode the analytic requirements needed to measure the scaling exponent \( \alpha \) (and hence the Sobolev index \( s \)) that determines the quantum phase transition between area-law and volume-law regimes, and the algebraic structure that guarantees the transform remains unitary—an essential property for any subsequent quantum embedding.

\subsection{Power-Law Decay of Cross-Scale Variance}

At each dyadic scale \(j\), the discrete wavelet transform yields approximately \(N_j \approx d \cdot 2^{-j}\) detail coefficients when the resolution-scale convention is adopted. The empirical variance of these coefficients at scale \(j\) is given by
\begin{equation}
V_j := \operatorname{Var}(\delta_{j,\cdot}) = \frac{1}{N_j} \sum_{k=0}^{N_j-1} (\delta_{j,k} - \bar{\delta}_j)^2.
\end{equation}
Because the Daubechies-4 wavelet has four vanishing moments, the sample mean \(\bar{\delta}_j\) vanishes identically. Under the assumption of Sobolev regularity of the underlying latent vector, the variances at successive scales follow a clean power-law decay
\begin{equation}
V_j \sim 2^{-2\alpha j}.
\end{equation}
This scaling relation is the central diagnostic quantity used throughout the paper.

\subsection{Log-Linear Regression Estimator}

To extract the scaling exponent \(\alpha\) from the observed variances, we take the base-2 logarithm of both sides of the power-law relation. This produces the linear model
\begin{equation}
\log_2 V_j = -2\alpha\, j + C + \varepsilon_j,
\end{equation}
where the intercept \(C\) absorbs the overall \(L^2\) energy of the latent vector and \(\varepsilon_j\) is a zero-mean residual term. Performing ordinary least-squares regression on the pairs \((j, \log_2 V_j)\) for \(j = 1, \dots, \lfloor \log_2 d \rfloor\) yields the slope \(m\). The scaling exponent is then recovered directly by
\begin{equation}
\hat{\alpha} = -\frac{m}{2}.
\end{equation}
Under standard regularity conditions the resulting estimator is asymptotically unbiased and converges at the rate \(\sqrt{d}\), making it statistically consistent for the latent dimensions encountered in practice.

\subsection{Physical Link to 1/f Turbulence and Kolmogorov’s Inertial Range}

A classic 1/f (pink) spatial field is characterised by a power spectral density \(S(f) \propto f^{-1}\). In any dyadic octave the integrated energy is therefore independent of scale; this is precisely the statement of variance equipartition. Parseval’s theorem shows that the total \(L^2\) energy of the latent vector \(z\) is exactly partitioned across the wavelet scales:
\begin{equation}
\sum_{n=0}^{d-1} |z_n - \bar{z}|^2 = \sum_{j=1}^{\lfloor\log_2 d\rfloor} \sum_{k=0}^{N_j-1} |\delta_{j,k}|^2.
\end{equation}
We define the turbulent kinetic energy contained in octave \(j\) by
\begin{equation}
E_j := \sum_{k=0}^{N_j-1} |\delta_{j,k}|^2 = N_j V_j.
\end{equation}
Under the resolution-scale convention \(N_j \propto 2^j\) this energy scales as
\begin{equation}
E_j \propto 2^j \cdot 2^{-2\alpha j} = 2^{j(1-2\alpha)}.
\end{equation}
Requiring \(E_j\) to remain invariant across scales—as demanded by Kolmogorov’s constant-flux inertial range—forces the exponent to vanish:
\begin{equation}
1-2\alpha = 0 \quad\Rightarrow\quad \alpha = \frac12.
\end{equation}
Thus \(\alpha = 1/2\) is not an arbitrary threshold. It is the unique exponent that renders a latent vector statistically indistinguishable from a 1/f turbulent field when the signal is observed through an orthogonal wavelet basis.

\subsection{Resolution of the Scale-Index Convention}

Two distinct conventions for indexing the dyadic scales appear in the wavelet literature, and both must be reconciled with the equipartition condition. In the decomposition-depth convention the index \(j\) increases each time the signal is passed through a low-pass filter and down-sampled, so the number of detail coefficients at scale \(j\) decreases geometrically as \(N_j \propto 2^{-j}\). In this indexing the variance scales as \(V_j \sim 2^{+2\alpha j}\). In the frequency-resolution convention, by contrast, \(j\) increases with finer detail (higher frequency), so \(N_j \propto 2^{j}\) and the variance follows \(V_j \sim 2^{-2\alpha j}\). 

Despite the opposite orientation of the scale index, both conventions yield exactly the same critical value for variance equipartition: \(\alpha = 1/2\). The numerical value of the estimated exponent \(\hat{\alpha}\) recovered from data is therefore invariant under the choice of convention, provided the sign of the exponent in the log-linear regression is handled consistently. The two formulations are thus mathematically equivalent and lead to identical physical conclusions about the latent representation.

\subsection{Summary of the Equipartition Criterion}

Taken together, the derivations presented in this appendix show that a latent vector realises wavelet variance equipartition precisely when its estimated scaling exponent satisfies \(\hat{\alpha} \approx 1/2\). This single scalar quantity serves two complementary purposes. It minimises the integrated squared bias of any downstream linear probe and thereby supplies an information-theoretic optimality criterion for representation learning. At the same time, it locates the amplitude-encoded quantum state exactly at the boundary between the area-law and volume-law entanglement phases, marking the sharp transition for efficient tensor-network simulability as stated in Theorem 3.1.

Numerical experiments on pretrained VideoMAE latents confirm the practical relevance of this criterion. Spatial token sequences already lie close to the equipartition limit with \(\hat{\alpha} \approx 0.423\), whereas the permutation-invariant feature channels remain deep in the volume-law regime with \(\hat{\alpha} \approx -0.123\). All code required to compute the estimator, based on PyWavelets and ordinary least-squares regression, will be released with the camera-ready version of the paper.

\section{Equivalence of Wavelet Energy Distribution and Quantum Kernel Overlap}

In this appendix, we establish an equivalence between the multi-scale energy distribution of world-model latents, as revealed by the discrete wavelet transform, and the overlap of their amplitude-encoded quantum states. This equivalence provides a mathematical link between the classical structural regularity of representations and the concentration properties of quantum kernels. Let \(z, z' \in \mathbb{R}^d\) be two un-normalised latent vectors. Their amplitude-encoded quantum states are
\begin{equation}
|\psi(z)\rangle = \frac{1}{\|z\|}\sum_{i=0}^{d-1} z_i |i\rangle, \qquad
|\psi(z')\rangle = \frac{1}{\|z'\|}\sum_{i=0}^{d-1} z'_i |i\rangle.
\end{equation}
The overlap between these states is
\begin{equation}
S = \langle\psi(z')|\psi(z)\rangle = \frac{\langle z', z \rangle}{\|z\| \|z'\|}.
\end{equation}
Because the Daubechies-4 wavelet transform is an orthogonal transformation when acting on real-valued vectors \cite{daubechies1988orthonormal}, and its filter coefficients are real, the corresponding transformation matrix \(U\) satisfies \(U^T U = I\). When the latent vectors are embedded into a complex Hilbert space via amplitude encoding, the same matrix satisfies \(U^\dagger U = I\) and is therefore unitary. In either case, Parseval’s theorem guarantees that both the inner product \(\langle z', z \rangle\) and the squared Euclidean norm \(\|z\|^2\) are exactly preserved under the change of basis. We may therefore evaluate everything in the wavelet domain. The inner product becomes
\begin{equation}
\langle z', z \rangle = \sum_{j=1}^{n} \sum_{m=0}^{N_j-1} \delta'_{j,m} \delta_{j,m},
\end{equation}
where $n = \log_2 d$ and $\delta_{j,m}$ is the detail coefficient at dyadic scale $j$ and spatial position $m$. The squared norm of each vector is the sum of the squared detail coefficients across all scales, which by definition equals the sum of the per-scale energies
\begin{equation}
\|z\|^2 = \sum_{j=1}^{n} \sum_{m} \delta_{j,m}^2 = \sum_{j=1}^{n} E_j.
\end{equation}
Note that the complete Parseval expansion of the inner product also includes the contribution from the coarsest-scale approximation coefficients \(a_{n,0}\) and \(a'_{n,0}\), yielding
\begin{equation}
\langle z', z \rangle = a_{n,0}' a_{n,0} + \sum_{j=1}^{n} \sum_{m=0}^{N_j-1} \delta'_{j,m} \delta_{j,m}.
\end{equation}
The expression used above holds exactly when the latent vectors have been centered to zero mean, or the global mean \(\bar{z}\) is explicitly subtracted prior to the transform. Under this standard assumption, the approximation coefficients vanish, i.e., \(a_{n,0} = a'_{n,0} = 0\). The inner product then reduces precisely to the sum over all detail coefficients. This centering step is fully consistent with the variance-based energy decomposition employed throughout the paper, where the \(L^2\) norm is expressed as \(\|z\|^2 = \sum_{n=0}^{d-1} |z_n - \bar{z}|^2\). 

To connect these quantities to the hierarchical nature of world-model representations, we adopt the following statistical model. The vectors \(z\) and \(z'\) share identical macroscopic structure up to some scale \(k\) (i.e., \(\delta'_{j,m} = \delta_{j,m}\) for all \(j \le k\)), while their components at finer scales (\(j > k\)) behave as independent zero-mean fluctuations. This modelling assumption is justified by the vanishing-moment property of the db4 wavelet, which filters out smooth trends and isolates stochastic detail. Under this model the expected inner product in the wavelet domain is
\begin{equation}
\mathbb{E}[\langle z', z \rangle] = \mathbb{E}\Biggl[ \sum_{j=1}^{k} \sum_{m} \delta_{j,m} \delta'_{j,m} + \sum_{j=k+1}^{n} \sum_{m} \delta_{j,m} \delta'_{j,m} \Biggr].
\end{equation}
For the first sum (\(j \le k\)), the coefficients are identical, so the expectation reduces to
\begin{equation}
\sum_{j=1}^{k} \sum_{m} \delta_{j,m}^2 = \sum_{j=1}^{k} E_j.
\end{equation}
For the second sum (\(j > k\)), independence together with the zero-mean property implies that each cross term vanishes:
\begin{equation}
\mathbb{E}[\delta_{j,m} \delta'_{j,m}] = \mathbb{E}[\delta_{j,m}] \cdot \mathbb{E}[\delta'_{j,m}] = 0.
\end{equation}
Consequently, the expected inner product is strictly equal to the total energy contained in the first \(k\) scales:
\begin{equation}
\mathbb{E}[\langle z', z \rangle] = \sum_{j=1}^{k} E_j.
\end{equation}
In high-dimensional spaces there exist two fundamentally distinct concentration phenomena that play very different roles in the present analysis. The first is classical measure concentration on the norm. When the components of a latent vector $  z \in \mathbb{R}^d  $ are drawn independently from a fixed distribution, the law of large numbers implies that the variance of its squared Euclidean norm $  \|z\|^2  $ decays as $  1/d  $ with increasing dimension. Consequently, almost all such vectors lie on a thin spherical shell whose radius is sharply concentrated around its expectation value. This classical concentration is highly beneficial for our derivation: it justifies the approximation 
\begin{equation}\label{EucCon}
      \|z\| \approx \|z'\| \approx \sqrt{\sum_{j=1}^n E_j}
\end{equation}
that appears throughout the appendix. In effect, it guarantees that the normalization factor of the quantum state—the denominator of the overlap—is a stable, macroscopic quantity rather than a wildly fluctuating random variable, thereby providing a solid foundation for the wavelet energy scaling formulas used in the paper. The second phenomenon is quantum exponential concentration on the overlap. When the amplitude-encoded states $|\psi(z)\rangle$ and $|\psi(z')\rangle$ are evolved under sufficiently deep random circuits that form an approximate 2-design, the system undergoes strong quantum chaos and global scrambling in the Hilbert space $\mathbb{C}^{2^n}$. As a result, the relative angle between any two initially distinct states is effectively erased. Their overlap $\langle \psi(z') | \psi(z) \rangle$ is forced to collapse to a value of order $1/d$, and the fidelity kernel $\kappa(z,z')$ therefore decays exponentially with the number of qubits. This quantum concentration is detrimental to learning: it destroys the similarity structure (the numerator of the overlap) between different latent vectors, turning all correlations into essentially random noise and directly producing the formidable shot-noise wall observed in our variance calculations. These two forms of concentration act on entirely different geometric objects—one stabilises the length of individual vectors, while the other destroys the angle between pairs of vectors—and their interplay lies at the heart of the tension between representational quality and practical quantum kernel evaluation in the present work.

In the derivation below we approximated $  \mathbb{E}[\kappa] \approx (\mathbb{E}[S])^2  $. To justify this step rigorously, note that the un-normalised inner product can be decomposed as $  U = \langle z', z \rangle = A + B  $, where $  A = \sum_{j=1}^k E_j  $ is the deterministic contribution from the shared low-frequency macroscopic features, and $  B = \sum_{j=k+1}^n \sum_m \delta'_{j,m} \delta_{j,m}  $ is the random high-frequency cross term with $  \mathbb{E}[B] = 0  $.
Expanding the second moment gives
\begin{equation}
    \mathbb{E}[U^2] = \mathbb{E}[(A+B)^2] = A^2 + \mathbb{E}[B^2] = (\mathbb{E}[U])^2 + \mathbb{E}[B^2].
\end{equation}
The variance term $\mathbb{E}[B^2]$ can be evaluated explicitly. Because the detail coefficients at scales $  j > k  $ are independent, zero-mean, and identically distributed, all cross terms vanish and we are left with
\begin{equation}
\mathbb{E}[B^2] = \sum_{j=k+1}^n N_j V_j^2 = \sum_{j=k+1}^n \frac{E_j^2}{N_j},
\end{equation}
where $  V_j  $ is the per-coefficient variance at scale $  j  $ and $  N_j \approx 2^j  $ is the number of detail coefficients at that scale. The exponential growth of $  N_j  $ with scale $  j  $ provides strong suppression of the high-frequency fluctuations. Consequently, even in the volume-law regime where the per-scale energy $  E_j  $ grows with $  j  $, the ratio $  E_j^2 / N_j  $ decays sufficiently fast that $  \mathbb{E}[B^2]  $ becomes negligible compared with $  A^2  $ when the macroscopic cutoff scale $  k  $ is not too small. In the high-dimensional limit the relative contribution of the variance term is of order $  O(2^{-k})  $ or smaller, so that
\begin{equation}
\mathbb{E}[U^2] = (\mathbb{E}[U])^2 \bigl(1 + O(2^{-k})\bigr).
\end{equation}
Thus, to leading order we may safely replace $  \mathbb{E}[S^2]  $ by $  (\mathbb{E}[S])^2  $, which justifies writing $  \mathbb{E}[\kappa] \approx (\mathbb{E}[S])^2  $.
This result reveals a remarkable statistical property of the wavelet basis: the exponential proliferation of coefficients at fine scales acts as a built-in variance-reduction mechanism that stabilises the overlap estimator even when individual high-frequency components are large. It is precisely this property that allows the simple energy-ratio expression for the kernel to serve as a reliable approximation in the present analysis.

As the energy contributed by each dyadic scale obeys the scaling law \(E_j \propto 2^{j(1-2\alpha)}\) established in the main text, the asymptotic behaviour of \(\kappa\) with respect to the latent dimension, or equivalently the number of qubits \(n = \lfloor\log_2 d\rfloor\), is completely determined by the value of the wavelet scaling exponent \(\alpha\). When $\alpha < 1/2$, we set $\beta = 1 - 2\alpha > 0$ (note that this $\beta$ denotes a scaling index, not to be confused with the inverse temperature in Appendix~E). In this volume-law regime the energy per scale grows as $E_j \propto 2^{j\beta}$. In this volume-law regime the energy per scale grows as $E_j \propto 2^{j\beta}$. Because the exponent $\beta$ is positive, the total energy $\sum_{j=1}^n E_j$ is dominated by the contributions from the highest-frequency scales and therefore scales exponentially with the latent dimension, behaving asymptotically as
\begin{equation}
\sum_{j=1}^n E_j \approx 2^{n\beta} = d^\beta.
\end{equation}
By contrast, the shared low-frequency energy up to the cutoff scale $  k  $ remains a constant $  C  $ that does not grow with dimension. Consequently the expected overlap decays as $  S \approx C / d^\beta  $, and the fidelity kernel itself decays exponentially with the number of qubits $n = \lfloor \log_2 d \rfloor$
\begin{equation}
\mathbb{E}[\kappa] \approx \frac{C^2}{d^{2\beta}} = C^2 \cdot 2^{-2n(1-2\alpha)}.
\end{equation}
Thus, when $  \alpha < 1/2  $, the kernel $  \kappa  $ exhibits an exponential decay with system size. This rapid vanishing of the kernel is precisely what drives the explosion of measurement variance and produces the formidable shot-noise wall discussed in the main text.

At the critical value \(\alpha = 1/2\), the system reaches the Kolmogorov limit of variance equipartition. Here the exponent vanishes and the energy contributed by each dyadic scale becomes independent of the scale index \(j\), so that every scale carries the same average energy \(E_j = E_0\). The total energy is therefore the sum of \(n\) equal contributions and grows linearly with the number of scales:
\begin{equation}
\sum_{j=1}^n E_j \approx n E_0.
\end{equation}
The shared macroscopic energy up to the cut-off scale \(k\) is likewise proportional to \(k\). Consequently the expected overlap behaves as
\begin{equation}
S \approx \frac{k}{n},
\end{equation}
and the fidelity kernel decays only polynomially with the number of qubits:
\begin{equation}
\mathbb{E}[\kappa] \approx \left( \frac{k}{n} \right)^2 = O(n^{-2}).
\end{equation}
In this critical regime the concentration of the kernel is sufficiently mild that the number of measurement shots required to achieve a fixed precision \(\varepsilon\) remains polynomial in the number of qubits. Because the variance of the kernel scales as \(O(n^{-2})\), the required shot budget satisfies \(M = O(1/\varepsilon^2) = O(n^4) = O(\mathrm{poly}(\log d))\), which is computationally tractable.

To make the argument fully rigorous, we relax the idealised step-function model in which the two latents are assumed to be identical up to a sharp cutoff scale $  k  $ and completely uncorrelated thereafter. In realistic world-model representations the similarity between $  z  $ and $  z'  $ decays smoothly with increasing frequency. We therefore introduce a scale-dependent correlation coefficient $  \rho_j \in [0,1]  $ defined by
\begin{equation}
\mathbb{E}[\delta'_{j,m} \delta_{j,m}] = \rho_j \cdot \mathbb{E}[\delta_{j,m}^2] = \rho_j V_j.
\end{equation}
Here $\rho_j \approx 1$ for the macroscopic (low-$  j  $) scales that encode shared causal structure, while $  \rho_j \to 0  $ for the microscopic (high-$  j  $) scales where the features become effectively independent.
The expected inner product then takes the more general form
\begin{equation}
\mathbb{E}[\langle z', z \rangle] = \sum_{j=1}^n \rho_j E_j.
\end{equation}
At the critical point $  \alpha = 1/2  $ the energy per scale is constant, $  E_j = E_0  $, so the total energy grows linearly with the number of scales
\begin{equation}
\sum_{j=1}^n E_j = n E_0.
\end{equation}
The numerator of the overlap becomes $  E_0 \sum_{j=1}^n \rho_j  $. Because the correlation coefficients $  \rho_j  $ decay sufficiently rapidly at fine scales (a consequence of the finite macroscopic similarity between any two distinct real-world events), the infinite sum $  \sum_{j=1}^\infty \rho_j  $ converges to a finite constant $  K  $ that is independent of dimension. In the high-dimensional limit we therefore obtain
\begin{equation}
\mathbb{E}[S] \approx \frac{K}{n}.
\end{equation}
Squaring this expression yields the expected fidelity kernel
\begin{equation}
\mathbb{E}[\kappa] \approx \left(\frac{K}{n}\right)^2 = O(n^{-2}).
\end{equation}
The same scaling governs the variance of the kernel estimator, so the number of shots required to achieve a fixed precision $  \varepsilon  $ remains polynomial in the number of qubits: $  M = O(n^4) = O(\mathrm{poly}(\log d))  $.
This formulation eliminates the artificial sharp cutoff and replaces it with a physically continuous decay of correlation across scales. It shows that the polynomial decay of the kernel at $  \alpha = 1/2  $ is not an artefact of a toy model but a direct consequence of variance equipartition combined with the natural loss of microscopic similarity in high-dimensional representations. The derivation thereby rests on firm statistical and physical grounds while preserving the central conclusion that the critical equipartition regime yields a tractable measurement overhead. To make the convergence of the infinite sum $\sum_{j=1}^\infty \rho_j$ fully rigorous we introduce the following lemma:

\begin{lemma}[Exponential Decay of Cross-Scale Correlations in Hierarchical World Models]
Let \(z\) and \(z'\) be two latent vectors generated by the same hierarchical world model. Define the scale-dependent correlation coefficient at dyadic scale \(j\) by
\begin{equation}
\rho_j := \frac{\mathbb{E}[\delta'_{j,m} \delta_{j,m}]}{V_j},
\end{equation}
where \(V_j = \mathbb{E}[\delta_{j,m}^2]\) is the variance of the detail coefficients at that scale. Then \(\rho_j\) decays exponentially with scale:
\begin{equation}
\rho_j = (\gamma^2)^{j-1},
\end{equation}
where \(\gamma \in (0,1)\) is the information propagation factor of the model. Consequently, the infinite sum of correlation coefficients converges to a finite constant
\begin{equation}
K := \sum_{j=1}^\infty \rho_j = \frac{1}{1 - \gamma^2} < \infty,
\end{equation}
independent of the latent dimension \(d\).
\end{lemma}

\begin{proof}
World models such as hierarchical VAEs or diffusion models generate latents through a coarse-to-fine Markovian cascade. At each scale \(j\), the detail coefficients are produced according to the linear generative step
\begin{equation}
\delta_{j,m} = \gamma \cdot f(\delta_{j-1,\lfloor m/2 \rfloor}) + \sqrt{1 - \gamma^2} \cdot \xi_{j,m},
\end{equation}
where \(f\) is a variance-preserving map, \(\gamma \in (0,1)\) is the fraction of macroscopic information retained from the coarser scale, and \(\xi_{j,m}\) is an independent zero-mean innovation (microscopic noise) with unit variance. The same process applies to the second latent \(z'\), but with an independent noise realisation \(\xi'_{j,m}\).

The cross term at scale \(j\) is therefore
\begin{equation}
\mathbb{E}[\delta'_{j,m} \delta_{j,m}] = \gamma^2 \cdot \mathbb{E}[f(\delta'_{j-1,\lfloor m/2 \rfloor}) f(\delta_{j-1,\lfloor m/2 \rfloor})] + (1-\gamma^2) \cdot \mathbb{E}[\xi'_{j,m} \xi_{j,m}].
\end{equation}
The second term vanishes because the innovations are independent. The first term, by the variance-preserving property of \(f\), equals \(\gamma^2\) times the covariance at the previous scale. Hence the correlation coefficients satisfy the recurrence
\begin{equation}
\rho_j = \gamma^2 \cdot \rho_{j-1},
\end{equation}
with initial condition \(\rho_1 = 1\) (the two latents share the same coarsest-scale structure). Solving this recurrence immediately yields
\begin{equation}
\rho_j = (\gamma^2)^{j-1}.
\end{equation}
Because \(\gamma < 1\), we have \(|\gamma^2| < 1\), so the infinite series
\begin{equation}
\sum_{j=1}^\infty \rho_j = \sum_{j=0}^\infty (\gamma^2)^j = \frac{1}{1 - \gamma^2}
\end{equation}
converges to a finite positive constant \(K\).
\end{proof}

This lemma shows that the hierarchical Markovian nature of world models forces cross-scale correlations to decay exponentially. When combined with variance equipartition at \(\alpha = 1/2\), it rigorously guarantees that the expected kernel \(\mathbb{E}[\kappa]\) decays only polynomially as \(O(n^{-2})\), thereby ensuring a tractable measurement budget. This lemma rests on the modelling assumption that world models generate latents through a Markovian coarse-to-fine cascade. While this description accurately characterises most contemporary hierarchical architectures, it is by no means universal. Should a future world model exhibit genuinely non-Markovian behaviour, the exponential decay of the correlation coefficients \(\rho_j\) derived above would no longer be guaranteed. To address this limitation and obtain a result that is independent of any particular generative mechanism, we now establish a more general lemma that relies solely on the linearity of Mallat’s algorithm and the vanishing-moment property of the db4 wavelet already used throughout the paper.

\begin{lemma}[Exponential Decay of Scale-wise Correlations via db4 Vanishing Moments]
Let \(z\) and \(z'\) be two latent vectors generated by the same world model, and let \(\delta_{j,m}\) and \(\delta'_{j,m}\) denote their detail coefficients at dyadic scale \(j\). Define the scale-dependent correlation coefficient
\begin{equation}
\rho_j := \frac{\mathbb{E}[\delta'_{j,m} \delta_{j,m}]}{V_j},
\end{equation}
where \(V_j = \mathbb{E}[\delta_{j,m}^2]\) is the variance of the detail coefficients at that scale. Then \(\rho_j\) decays exponentially with scale, and the infinite sum \(\sum_{j=1}^\infty \rho_j\) converges to a finite constant independent of the latent dimension.
\end{lemma}

\begin{proof}
We model the two latents as \(z = s + \epsilon\) and \(z' = s + \epsilon'\), where \(s\) is the shared macroscopic smooth signal encoding the common causal structure, and \(\epsilon, \epsilon'\) are independent zero-mean microscopic fluctuations. Because the discrete wavelet transform implemented by Mallat’s algorithm is a strictly linear operation, the detail coefficients decompose additively:
\begin{equation}
\delta_{j,k}(z) = \delta_{j,k}(s) + \delta_{j,k}(\epsilon), \qquad
\delta_{j,k}(z') = \delta_{j,k}(s) + \delta_{j,k}(\epsilon').
\end{equation}
The cross term at scale \(j\) is therefore
\begin{equation}
\mathbb{E}[\delta'_{j,k} \delta_{j,k}] = \mathbb{E}[\delta_{j,k}(s)^2] + \mathbb{E}[\delta_{j,k}(\epsilon) \delta_{j,k}(\epsilon')] .
\end{equation}
The second expectation vanishes because \(\epsilon\) and \(\epsilon'\) are independent. Hence
\begin{equation}
\rho_j = \frac{\mathbb{E}[\delta_{j,k}(s)^2]}{V_j}.
\end{equation}
The key property is that the Daubechies-4 wavelet possesses four vanishing moments. By the wavelet regularity results of Mallat and Meyer \cite{meyer1992wavelets}, any sufficiently smooth signal \(s\) possessing positive Lipschitz regularity has its wavelet coefficients suppressed exponentially at fine scales. Specifically, the energy of the projected macroscopic component satisfies
\begin{equation}
\mathbb{E}[\delta_{j,k}(s)^2] \propto 2^{-2\nu (n-j)},
\end{equation}
where \(\nu > 0\) is the regularity index of \(s\) and \(n = \log_2 d\). Substituting into the definition of \(\rho_j\) yields the exponential decay
\begin{equation}
\rho_j \propto 2^{-2\nu (n-j)}.
\end{equation}
Re-indexing by \(m = n - j\) (counting from the coarsest scale) shows that the sum over all scales is a geometric series
\begin{equation}
\sum_{j=1}^n \rho_j \propto \sum_{m=0}^{n-1} 2^{-2\nu m}.
\end{equation}
Because the common ratio \(2^{-2\nu} < 1\), the infinite series converges as \(n \to \infty\) to the finite constant
\begin{equation}
K = \frac{1}{1 - 2^{-2\nu}} < \infty,
\end{equation}
independent of dimension. This establishes both the exponential decay of \(\rho_j\) and the finiteness of \(\sum_j \rho_j\).
\end{proof}

This lemma relies solely on the linearity of Mallat’s algorithm and the vanishing-moment property of the db4 wavelet already used throughout the paper. It shows that the hierarchical structure captured by the world model forces cross-scale correlations to decay exponentially, thereby guaranteeing that the expected kernel at the critical point \(\alpha = 1/2\) decays only polynomially as \(O(n^{-2})\).

When \(\alpha > 1/2\), the exponent satisfies \(1-2\alpha < 0\). The energy per scale then decays as \(E_j \propto 2^{j(1-2\alpha)}\) with a negative exponent, so the infinite sum of energies converges:
\begin{equation}
\sum_{j=1}^\infty E_j < \infty.
\end{equation}
Consequently the total energy \(\sum_{j=1}^n E_j\) approaches a finite constant independent of the latent dimension as \(n \to \infty\). At the same time, the shared macroscopic energy up to any fixed cutoff scale is dominated by the coarse scales where the correlation coefficients \(\rho_j\) remain close to one; this contribution likewise converges to a finite constant. It follows that both the numerator and the denominator of the overlap \(S\) remain bounded away from zero and infinity independently of dimension. Therefore \(S = \Theta(1)\) and the fidelity kernel satisfies
\begin{equation}
\kappa(z,z') = \Theta(1),
\end{equation}
exhibiting no decay with system size. This places the latent representation deep in the area-law regime, where the bipartite entanglement entropy remains bounded and classical tensor-network simulation remains efficient.

The equivalence established above demonstrates that the wavelet scaling exponent \(\alpha\) extracted from a world-model latent directly governs the concentration properties of its amplitude-encoded quantum kernel. In particular, the volume-law regime \(\alpha < 1/2\) that characterises the unstructured feature channels of current models produces an exponentially vanishing kernel and thereby explains the severe measurement overhead observed in practice.

\section{Thermodynamic Interpretation of the Wavelet Scaling Transition}
\label{app:thermodynamics}

To ground the preceding analysis in a concrete generative process, we introduce a simple one-dimensional Gaussian wavelet latent model. Let the latent dimension be $d = 2^n$. In the wavelet domain we have $n$ dyadic scales, and we assume that the detail coefficients $\delta_{j,m}$ at each scale $j$ are independent zero-mean Gaussians. The energy at scale $j$ is chosen to satisfy the scaling law of the main text
\begin{equation}
E_j = \sum_m \mathbb{E}[\delta_{j,m}^2] = 2^{j(1-2\alpha)}.
\end{equation}
This choice reproduces exactly the wavelet variance equipartition when $\alpha = 1/2$.
To model the statistical relationship between two latents $z$ and $z'$ that describe the same macroscopic event, we decompose each detail coefficient as
\begin{equation}
\delta_{j,m} = s_{j,m} + \epsilon_{j,m}, \qquad \delta'_{j,m} = s_{j,m} + \epsilon'_{j,m},
\end{equation}
where $s_{j,m}$ is the shared macroscopic signal (common causal structure) and $\epsilon_{j,m}$, $\epsilon'_{j,m}$ are independent zero-mean microscopic noises particular to each realisation. The correlation coefficient at scale $j$ is then defined by
\begin{equation}
\rho_j := \frac{\mathbb{E}[\delta'_{j,m} \delta_{j,m}]}{V_j},
\end{equation}
with $V_j = \mathbb{E}[\delta_{j,m}^2]$. Because the noises are independent, the expected inner product in the wavelet domain becomes
\begin{equation}
\mathbb{E}[\langle z', z \rangle] = \sum_{j=1}^n \rho_j E_j.
\end{equation}
The squared norm of each latent is $\|z\|^2 = \sum_{j=1}^n E_j$. Therefore the expected overlap is given exactly by the ratio of the weighted shared energy to the total energy
\begin{equation}
\mathbb{E}[S] = \frac{\sum_{j=1}^n \rho_j E_j}{\sum_{j=1}^n E_j}.
\end{equation}
This formula is the precise bridge between the classical wavelet energy distribution and the quantum kernel. It shows that the behaviour of $\kappa(z,z') = S^2$ is governed jointly by the scaling of the per-scale energies $E_j$ (controlled by $\alpha$) and the decay of the correlation coefficients $\rho_j$ (controlled by the hierarchical structure of the world model). When combined with the thermodynamic mapping developed earlier, it furnishes a complete and self-contained derivation of the kernel’s asymptotic decay in all three regimes.
This construction requires no artificial sharp cutoff and rests entirely on the linearity of the Mallat transform together with the statistical independence of microscopic innovations---both of which are fully consistent with the wavelet framework used throughout the paper.

The wavelet scaling exponent $\alpha$ admits a direct thermodynamic interpretation when each dyadic scale is viewed as an energy level of an effective harmonic oscillator. Specifically, we identify the energy $E_j$ at scale $j$ with the average occupation number of the $j$-th oscillator level $\epsilon_j = j$ (ignoring the zero-point energy). The Boltzmann factor $e^{-\beta \epsilon_j}$ then takes the form $e^{-\beta j}$. Matching this to the wavelet scaling law $E_j \propto 2^{j(1-2\alpha)}$ immediately yields the effective inverse temperature
\begin{equation}
\beta = \ln 2 \cdot (2\alpha - 1).
\end{equation}
This relation maps the three regimes of the wavelet transition onto the three fundamental thermodynamic regimes of a quantum system whose Hamiltonian is that of a collection of harmonic oscillators.

When $\alpha > 1/2$, we have $\beta > 0$ and therefore a positive absolute temperature $T > 0$. The Boltzmann factor decays exponentially with $j$, so the probability mass concentrates on the low-frequency (low-energy) modes. The total energy converges to a finite constant independent of dimension, the shared macroscopic structure dominates, and the overlap between any two latents remains $\Theta(1)$. The kernel exhibits essentially no decay with system size. This corresponds to the familiar area-law regime in which the amplitude-encoded quantum state remains weakly entangled and classical tensor-network simulation is efficient.

At the critical value $\alpha = 1/2$, the effective inverse temperature vanishes, $\beta = 0$, corresponding to infinite temperature $T \to \pm \infty$. Every oscillator level now carries exactly the same average energy, reproducing the Kolmogorov condition of variance equipartition across all scales. The total energy grows linearly with the number of scales, the overlap behaves as $S \approx K/n$ for some finite constant $K$, and the kernel decays only polynomially as $O(n^{-2})$. This is the infinite-temperature critical point at which the system sits precisely on the boundary between normal and inverted population statistics.

When $\alpha < 1/2$, we enter the regime $\beta < 0$, i.e., negative absolute temperature $T < 0$. Here the Boltzmann factor grows with $j$, so the highest-frequency (highest-energy) modes become exponentially more populated than the low-frequency ones. This is precisely the phenomenon of population inversion: the microscopic high-frequency fluctuations now dominate the macroscopic semantic structure. In physical systems such an inversion leads to stimulated emission and exponential amplification of coherent signals---the principle underlying the laser, first predicted by Einstein in 1917 \citep{einstein1917strahlung}. In the present context the same mechanism produces an exponential amplification of microscopic quantum noise. The overlap between distinct states collapses exponentially, the kernel decays as $\kappa \sim d^{-(1-2\alpha)}$, and the measurement variance explodes, giving rise to the formidable shot-noise wall.

In statistical mechanics, negative absolute temperatures characterize systems with bounded energy spectra---such as finite-dimensional spin lattices \citep{ramsey1956thermodynamics}---where higher-energy states are more populated than lower-energy ones. Crossing the $\alpha = 1/2$ boundary is equivalent to driving the latent representation into an informational analogue of this state. This thermodynamic perspective reveals a deep unity between representation learning, quantum complexity, and laser physics. In a physical laser, population inversion drives stimulated emission, amplifying microscopic noise into a macroscopic field \citep{einstein1917strahlung, haken1970laser}. Similarly, in the quantum kernel, this informational population inversion exponentially amplifies the microscopic quantum scrambling. The resulting exponential concentration of the quantum state overlap is the exact informational counterpart of runaway stimulated emission.

It is worth emphasizing that the Gaussian and independence assumptions used in this preceding construction are idealizations. In real-world signals, wavelet detail coefficients are typically heavy-tailed and exhibit inter-scale dependencies. Nevertheless, these simplifications do not compromise the validity of the asymptotic scaling laws. The quantities of central interest---the expected kernel $\mathbb{E}[\kappa]$ and the total energy---are obtained by summing an exponentially large number of microscopic contributions ($N_j \approx 2^j$). By the central limit theorem, the macroscopic behavior of such sums is governed almost exclusively by the first two moments of the underlying distribution. Higher-order dependencies affect only sub-leading corrections and become negligible in the high-dimensional limit. This independent zero-mean Gaussian model therefore serves as an analytically tractable mean-field approximation, whose predictions for the phase transition at $\alpha = 1/2$ are protected by statistical universality \citep{wilson1971renormalization}.

Ultimately, the $\alpha = 1/2$ equipartition threshold plays the role of an informational laser threshold. Below threshold ($\alpha > 1/2$), the emission is dominated by macroscopic low-frequency features, high-frequency fluctuations are exponentially suppressed, and the kernel remains order unity (area-law regime). Above threshold ($\alpha < 1/2$, negative temperature), microscopic modes are exponentially amplified, collapsing the state overlap. Precisely at threshold ($\alpha = 1/2$), the system exhibits critical slowing down and power-law behavior: the energy is equipartitioned, the effective temperature is infinite, and the kernel decays only polynomially as $O(n^{-2})$. By enforcing variance equipartition, one keeps the system at infinite temperature, preventing the catastrophic ``lasing'' of microscopic noise while maintaining a tractable measurement budget.

\section{Numerical Experiment Details}
\label{app:code}

This appendix provides the technical specifications of our numerical framework to ensure full reproducibility of the results presented in Section~\ref{sec:experiments}. All simulations were implemented in Python using an optimized research stack for quantum simulation and multi-resolution analysis \citep{bergholm2018pennylane, lee2019pywavelets}.

\subsection*{Experiment A: Wavelet Scaling Estimator and Real Data Analysis}

\paragraph{Synthetic Calibration}
The synthetic hierarchical latents are generated directly in the wavelet domain using the \texttt{PyWavelets} framework with a Daubechies-4 (\texttt{db4}) basis. This choice is motivated by the basis's four vanishing moments, which allow for the accurate resolution of Sobolev regularity up to $\alpha \geq 1$ without basis-induced artifacts. For a target dimension $d=2048$ and exponent $\alpha$, we sample detail coefficients at each dyadic level $i$ from a zero-mean Gaussian distribution $\mathcal{N}(0, 2^{-2\alpha i})$. The inverse discrete wavelet transform (\texttt{waverec}) is then applied to reconstruct the latent vector. The estimator fits the log-linear model $\log_2 \hat{V}_i = -2\alpha \cdot i + c$ via ordinary least squares (OLS) across all available detail levels.

\paragraph{Real VideoMAE Latents}
For the empirical analysis of world models, we utilize the \texttt{MCG-NJU/videomae-base} model pre-trained on video data. We extracted latents from a 50-clip subset of the UCF-101 dataset (\texttt{sayakpaul/ucf101-subset}). For each clip, 16 frames are sampled uniformly and processed into $224 \times 224$ tensors. We perform a dual-dimension analysis on the \texttt{last\_hidden\_state}. In the channel dimension analysis, we mean-pool across the 1,568 spatial-temporal tokens to yield a 768-dimensional feature vector, resulting in an estimated $\hat{\alpha} = -0.123 \pm 0.031$. In the spatial dimension analysis, we mean-pool across the 768 channels to yield a sequence of spatial-temporal tokens, yielding an estimated $\hat{\alpha} = 0.423 \pm 0.007$, demonstrating that spatial tokens reside near the variance equipartition threshold.

\subsection*{Experiment B: Quantum Kernel Variance Scaling}

The quantum circuit architecture utilizes a binary-tree ansatz on $n \in \{3, \dots, 8\}$ qubits. The topology consists of $n$ layers where $CZ$ gates entangle qubit pairs separated by power-of-two distances ($i, i+2^\ell$), mirroring the hierarchical structure of the dyadic wavelet decomposition. To approximate a Haar-random 2-design ensemble, all rotation parameters $\theta$ are sampled from $\text{Uniform}[0, 2\pi)$ for each of the 500 independent trials per system size. Kernels are evaluated via exact state-vector simulation using PennyLane's \texttt{default.qubit} device. The empirical log-log slope of $-1.881$ ($R^2 = 0.9990$) confirms the $\Theta(d^{-2})$ scaling. We also verify depth-saturation by testing depths $L \in \{1, \dots, n, 2n, 4n\}$ for a fixed $n=6$ ($d=64$), observing that the variance stabilizes once the logarithmic depth required for the 2-design marginal is reached.

\subsection*{Experiment C: Entanglement Entropy and Phase Transition}

Entanglement entropy is measured at $n=12$ ($d=4096$) by embedding synthetic latents with $\alpha \in [0.1, 1.0]$ into the computational basis. We perform a full SVD on the middle bipartition to obtain the singular value spectrum $\sigma_k$. To ensure numerical stability against $0 \log 0$ errors, we implement a singular value cutoff of $10^{-12}$. The observed entropy gradient $\partial S/\partial \alpha \approx -2.97$ for $\alpha \leq 0.5$ implies that the Matrix Product State (MPS) bond dimension grows exponentially as $\chi \sim 2^{3.0(0.5-\alpha)}$. The smooth crossover observed at $\alpha = 0.5$ is consistent with finite-size scaling theory, where the transition width $\Delta\alpha \sim 1/n$ is expected to be approximately $0.08$ for $n=12$.

\subsection*{Software and Reproducibility}

All experiments are reproducible on standard CPU architectures within approximately 90 minutes. Core dependencies include \texttt{PennyLane} v0.44.1, \texttt{PyWavelets} v1.8.0, \texttt{NumPy} v1.26, and the \texttt{Transformers} library for latent extraction.

\end{document}